\documentclass[sigconf, nonacm]{acmart}
\usepackage{amsmath}
\usepackage{xcolor}
\usepackage{colortbl}
\usepackage{multirow}
\usepackage{todonotes}
\usepackage[ruled,vlined,linesnumbered]{algorithm2e}

\usepackage[noend]{algpseudocode}
\usepackage{makecell}
\usepackage{caption}
\usepackage{subcaption}
\usepackage{amsthm}    % 提供 theorem 环境
\usepackage{hyperref}  % 提供超链接功能
\usepackage{xspace}
\newcommand{\system}{\textsf{ForeSight}\xspace}
\newcommand{\aspn}{ASPN\xspace}
\newcommand{\mtfs}{MTFS\xspace}

\newtheorem{ruledef}{Rule}
\newtheorem{theorem}{Theorem}

\algrenewcommand\algorithmicrequire{\textbf{Input:}}
\algrenewcommand\algorithmicensure{\textbf{Output:}}

\newcommand\vldbdoi{XX.XX/XXX.XX}
\newcommand\vldbpages{XXX-XXX}
\newcommand\vldbvolume{14}
\newcommand\vldbissue{1}
\newcommand\vldbyear{2020}
\newcommand\vldbauthors{\authors}
\newcommand\vldbtitle{\shorttitle} 
\newcommand\vldbavailabilityurl{URL_TO_YOUR_ARTIFACTS}
\newcommand\vldbpagestyle{plain} 

\newcommand{\showDOI}[1]{\unskip}

\begin{document}
% \sloppy

\title{ForeSight: A Predictive-Scheduling Deterministic Database}

% Junfang Huang, Yu Yan, Hongzhi Wang, Yingze Li, Jinghan Lin
\author{Junfang Huang}
\affiliation{%
    \institution{Harbin Institute of Technology}
    \streetaddress{92 West Dazhi St}
    \city{Harbin}
    \state{Heilongjiang}
    \country{China}
}
\email{24S103299@stu.hit.edu.cn}

\author{Yu Yan}
\affiliation{%
    \institution{Harbin Institute of Technology}
    \streetaddress{92 West Dazhi St}
    \city{Harbin}
    \state{Heilongjiang}
    \country{China}
}
\email{yuyan@hit.edu.cn}

\author{Hongzhi Wang}
\affiliation{%
    \institution{Harbin Institute of Technology}
    \streetaddress{92 West Dazhi St}
    \city{Harbin}
    \state{Heilongjiang}
    \country{China}
}
\email{wangzh@hit.edu.cn}

\author{Yingze Li}
\affiliation{%
    \institution{Harbin Institute of Technology}
    \streetaddress{92 West Dazhi St}
    \city{Harbin}
    \state{Heilongjiang}
    \country{China}
}
\email{23B903046@stu.hit.edu.cn}

\author{Jinghan Lin}
\affiliation{%
    \institution{Harbin Institute of Technology}
    \streetaddress{92 West Dazhi St}
    \city{Harbin}
    \state{Heilongjiang}
    \country{China}
}
\email{linjinghan@stu.hit.edu.cn}

% Deterministic databases enable scalable replicated systems by executing transactions in a predetermined order. However, existing designs fail to capture transaction dependencies, leading to insufficient scheduling, high abort rates, and poor resource utilization. By addressing these challenges with lightweight conflict prediction and informed scheduling, we present ForeSight, a high-performance deterministic database system. Our system has three core improvements: (1) We design an Association Sum-Product Network to predict potential transaction conflicts, providing the input for dependency analysis without pre-obtained read/write sets. (2) We enhance the storage engine to integrate multi-version-based optimization, improving the execution process and fallback strategy to boost commit rates and concurrency. (3) We propose a matrix two-pass forward scan algorithm that performs dependency analysis to generate conflict-aware schedules, significantly reducing scheduling overhead. Experimental results on multiple benchmarks show that ForeSight achieves up to 2 times higher throughput on skewed workloads and maintains strong performance under contention, demonstrating that predictive scheduling substantially improves deterministic database scalability. 
\begin{abstract}
  Deterministic databases enable scalable replicated systems by executing transactions in a predetermined order. However, existing designs fail to capture transaction dependencies, leading to insufficient scheduling, high abort rates, and poor resource utilization. 
  By addressing these challenges with lightweight conflict prediction and informed scheduling, we present \system, a high-performance deterministic database system. Our system has three core improvements: 
  (1) We design an Association Sum-Product Network to predict potential transaction conflicts, providing the input for dependency analysis without pre-obtained read/write sets.
  (2) We enhance the storage engine to integrate multi-version-based optimization, improving the execution process and fallback strategy to boost commit rates and concurrency.
  (3) We propose a matrix two-pass forward scan algorithm that performs dependency analysis to generate conflict-aware schedules, significantly reducing scheduling overhead.
  Experimental results on multiple benchmarks show that \system achieves up to 2 times higher throughput on skewed workloads and maintains strong performance under contention, demonstrating that predictive scheduling substantially improves deterministic database scalability.
\end{abstract}

\maketitle

\pagestyle{\vldbpagestyle}
\begingroup\small\noindent\raggedright\textbf{PVLDB Reference Format:}\\
\vldbauthors. \vldbtitle. PVLDB, \vldbvolume(\vldbissue): \vldbpages, \vldbyear.\\
\href{https://doi.org/\vldbdoi}{doi:\vldbdoi}
\endgroup
\begingroup
\renewcommand\thefootnote{}\footnote{\noindent
This work is licensed under the Creative Commons BY-NC-ND 4.0 International License. Visit \url{https://creativecommons.org/licenses/by-nc-nd/4.0/} to view a copy of this license. For any use beyond those covered by this license, obtain permission by emailing \href{mailto:info@vldb.org}{info@vldb.org}. Copyright is held by the owner/author(s). Publication rights licensed to the VLDB Endowment. \\
\raggedright Proceedings of the VLDB Endowment, Vol. \vldbvolume, No. \vldbissue\ %
ISSN 2150-8097. \\
\href{https://doi.org/\vldbdoi}{doi:\vldbdoi} \\
}\addtocounter{footnote}{-1}\endgroup

\ifdefempty{\vldbavailabilityurl}{}{
\vspace{.3cm}
\begingroup\small\noindent\raggedright\textbf{PVLDB Artifact Availability:}\\
The source code, data, and/or other artifacts have been made available at \url{https://github.com/AvatarTwi/foresight_system}.
\endgroup
}

\section{Introduction}
\label{sec:intro}

Deterministic Concurrency Control (DCC) offers a simple and effective way to build distributed databases that are both highly available and scalable~\cite{DCC,abadi2018overview,wang2025comprehensive,ren2014evaluation,han2025optimized}. The major idea is to ensure that different replicas consistently produce the same results, as long as the same input transactions are provided. 
Executing transactions in a predetermined order eliminates runtime nondeterminism and avoids complex conflict resolution mechanisms such as locking or coordination protocols. 
DCC offers several practical advantages. (1) It improves execution efficiency by avoiding distributed locking and synchronous replication. (2) It reduces network overhead by controlling transaction order instead of propagating data tuple changes. (3) It simplifies failure recovery, as a consistent state can be restored by replaying input logs without additional coordination. These features make DCC a promising foundation for distributed databases that aim to combine determinism, high concurrency, and fault tolerance, and have inspired adoption in modern systems such as OpenGauss~\cite{li2021opengauss}.

Specifically, existing deterministic databases can be broadly categorized into pessimistic and optimistic approaches~\cite{DCC,abadi2018overview,occ,wang2016mostly,Tictoc}. Pessimistic deterministic databases, including BOHM~\cite{BOHM}, PWV~\cite{PWV}, and Calvin~\cite{calvin}, enforce a fixed transaction order before execution, typically through dependency graphs or ordered locks, to ensure consistent execution across replicas. For example, BOHM builds a dependency graph based on pre-obtained read and write sets to determine execution order. Calvin enforces a global order of transaction execution using lock-based scheduling to achieve deterministic concurrency control. Overall, pessimistic databases require prior knowledge of the read/write sets of transactions and typically involve at least two rounds of execution. In contrast, optimistic deterministic databases, such as Aria~\cite{aria} and AriaER~\cite{ariaER}, execute batches of transactions in parallel on a shared snapshot without prior knowledge of read and write sets. Then, optimistic approaches ensure consistent commit status through deterministic validation and often incorporate a reordering mechanism to maximize the number of successful commits within a batch.

Although existing deterministic databases have demonstrated the potential for efficient and scalable transaction processing, they still face three key challenges that limit the performance under high concurrency and skewed access patterns.

% hjf: 缩短了一下，突出问题的核心，避免赘述，收束至：缺乏前置冲突认知 → 执行期才“碰撞” → 高成本
(1)~\textbf{Costly requirement of read/write sets for conflict detection.} Pessimistic approaches~\cite{BOHM, PWV, calvin} rely on pre-obtained read/write sets to construct dependency graphs and detect transaction conflicts. However, acquiring accurate read/write sets often requires pre-execution, which is costly and complicates deployment in real production environments.

(2)~\textbf{High abort rates under contention.} Optimistic approaches adopt an ``execute-then-validate'' strategy~\cite{aria}, running transactions speculatively on a snapshot. However, because conflicts are only detected after execution, workloads with high contention or skewed access patterns can lead to excessive aborts, wasting both computation and memory resources.

(3)~\textbf{Ineffective reordering for transaction scheduling.} Existing deterministic databases schedule transactions inefficiently due to ineffective reordering strategies. Optimistic deterministic databases~\cite{aria} attempt to reorder transactions using coarse rules (e.g., aborting transactions that depend on earlier ones) without analyzing whether dependencies actually form cycles, leading to unnecessary aborts. Pessimistic deterministic databases~\cite{BOHM,PWV} build fine-grained dependency graphs, but their graph-based reordering is computationally expensive, making them impractical at large batch sizes or high throughput.

%%%%%这里建议分析一下这几个问题的核心来源，从核心没有解决的问题引出本文方法的动机，这样体现出本文方法动机的特点
% hjf: 分析核心来源为缺乏前置冲突认知, 引出本文方法的动机为预测式调度，再让三项技术对应三个问题
These challenges all stem from the absence of predictive conflict analysis before execution. If potential conflicts could be captured before execution, systems could avoid relying on execution to obtain read/write sets, eliminate redundant aborts, and reduce scheduling overhead.

Motivated by this, we propose \system, a deterministic database that introduces \emph{predictive scheduling} as its central idea. With foresight into conflicts, our system transforms blind execution into conflict-aware scheduling. \system incorporates three complementary techniques:

%%%%%建议此处对应上面三个问题来写，特别是从技术角度能逻辑通顺地解决，而不是自己加描述
% hjf: 对应上面三个问题来写
(1)~\textbf{Conflict prediction with \aspn.} To eliminate the costly requirement for pre-obtained read/write sets, \system builds an \emph{Association Sum-Product Network (ASPN)} over table relations to capture attribute correlations and access patterns, enabling fast, accurate conflict prediction without pre-executing transactions.

(2)~\textbf{Multi-version-based optimization for reduced aborts.} 
Through multi-version-based optimization, \system both relaxes commit constraints and enhances the fallback strategy, allowing safe commits under version conflicts and more effective handling of residual conflicts, which together significantly improve commit rates under contention.

(3)~\textbf{Lightweight reordering for efficient scheduling.} To improve the efficiency of transaction scheduling, we propose \emph{Matrix Two-pass Forward Scan (\mtfs)}, a lightweight reordering algorithm designed to support efficient scheduling. \mtfs performs two scans to construct a simplified dependency structure, enabling the system to identify cycle-free schedules with minimal overhead.

% \footnote{summarize the contributions here}
Together, these three techniques enable \emph{predictive scheduling}, addressing the key challenges of existing deterministic databases.

\noindent\textbf{Contributions.} This paper makes the following contributions: 

\begin{itemize}
    \item We explore \emph{predictive scheduling} in deterministic databases, where it specifically means having foresight of transaction conflicts before execution to guide scheduling (Section~\ref{sec:overview}). 

    \item We integrate three complementary techniques to enable predictive scheduling: a lightweight conflict predictor \textbf{\aspn} (Section~\ref{sec:ACP}), a multi-version-based optimization to reduce aborts (Section~\ref{sec:optimization}), and an efficient reordering algorithm \textbf{\mtfs} for effective scheduling (Section~\ref{sec:reorder}).

    \item Extensive evaluations demonstrate that \system consistently outperforms state-of-the-art deterministic databases. \system achieves up to $2\times$ throughput improvement on skewed YCSB workloads and maintains superior stability under high-conflict TPC-C workloads (Section~\ref{sec:eval}). 
\end{itemize}

% The rest of this paper is organized as follows. Section~\ref{sec:background} provides background on deterministic databases. Section~\ref{sec:overview} gives an overview of the \system database. Section~\ref{sec:ACP} introduces the Association Sum-Product Network (\aspn) for dependency prediction. Section~\ref{sec:optimization} presents the multi-version storage design and deterministic commit strategy. Section~\ref{sec3:mtfs} describes the Matrix Two-pass Forward Scan (\mtfs) reordering algorithm for efficient scheduling. Section~\ref{sec:eval} presents extensive experimental results using a variety of benchmarks and workload settings, demonstrating the performance, scalability, and robustness of \system. Section~\ref{sec:related} discusses related work, and Section~\ref{sec:conclusion} concludes the whole paper.% The appendix provides additional implementation details.

\section{BACKGROUND}
\label{sec:background}
In this section, we introduce the typical architecture of deterministic databases by examining representative implementations.
Deterministic databases typically follow a five-layer architecture~\cite{aria,calvin, PWV, BOHM,abadi2018overview}: input, sequencing, scheduling, execution, and storage. These layers enable deterministic execution via message passing and coordinated interactions. Then, we will discuss the design principles and limitations of existing works from the perspective of the above five layers. 

% Table~\ref{tab:system-comparison} summarizes representative systems and their design choices.

% \begin{table}[t]
%   \caption{Comparison of Representative Deterministic Database Systems}
%   \label{tab:system-comparison}
%   \begin{tabular}{lcccc}
%     \toprule
%     \textbf{System}  & \textbf{Scheduling} & \textbf{Execution}  \\
%     \midrule
%     BOHM~\cite{BOHM}   & Placeholder Graph     & Parallel + Version Waiting        \\
%     PWV~\cite{PWV}    & Partitioned Graph & Shard-parallel, Early visibility \\
%     Calvin~\cite{calvin} & Ordered Locking     & Execute after Locking             \\
%     Aria~\cite{aria}   & Validation Retry & Snapshot Execution  \\
%     \bottomrule
%   \end{tabular}
% \end{table}

\noindent\textbf{Input Layer.}
Existing deterministic databases require the input layer to supply transactions in a form that meets the protocol’s determinism constraints~\cite{abadi2018overview}. Many approaches rely on predefined logic transactions with fully known execution behavior~\cite{BOHM,PWV,calvin}, enabling pre-obtained read/write sets and pre-scheduling, but limiting support for interactive or ad-hoc SQL transactions~\cite{shanbhag2017robust}.

\noindent\textbf{Sequencing Layer.}
The sequencing layer assigns a global order to all transactions to eliminate nondeterministic behavior (e.g., randomness, timestamps) and ensure consistent results across replicas. Most existing systems simply order transactions without leveraging dependencies or reducing contention~\cite{aria,calvin,PWV,BOHM,abadi2018overview}, deferring conflict handling to later phases and risking higher abort rates under skewed or high-contention workloads.

\noindent\textbf{Scheduling Layer.}
The scheduling layer constructs dependency relationships based on the global order and coordinates concurrent execution. Most systems adopt dependency-graph-based~\cite{BOHM,PWV}, ordered-locking-based~\cite{calvin}, or validation-based~\cite{aria} approaches, each with trade-offs in concurrency and overhead. For example, BOHM builds static dependency graphs with placeholders for writes and uses Multi-Version Concurrency Control (MVCC) to decouple reads from writes~\cite{BOHM}; PWV partitions transactions, builds per-partition dependency graphs, and enables early commit~\cite{PWV}; Calvin acquires locks in a predefined order to follow the global sequence~\cite{calvin}; and Aria executes transactions in parallel on a snapshot, followed by global validation and optional lock-based fallback~\cite{aria}. However, coarse-grained validation can cause excessive aborts under contention~\cite{aria}, while fine-grained graphs incur high construction costs for large batches~\cite{BOHM,PWV,calvin}.

\noindent\textbf{Execution Layer.}
The execution layer processes transactions following the dependencies determined by the scheduler. For example, BOHM executes in parallel using placeholders with version waiting~\cite{BOHM}; PWV processes shards in parallel with early write visibility~\cite{PWV}; Calvin executes only after acquiring all required locks~\cite{calvin}; and Aria runs transactions on a snapshot with batch validation and retry~\cite{aria}. However, under high contention, snapshot-based execution can suffer from high abort rates, while other approaches may lead to deadlocks or excessive waiting~\cite{aria,calvin,PWV,BOHM,Tictoc,diaconu2013hekaton}.

\noindent\textbf{Storage Layer.}
The storage model is closely tied to concurrency control, with multi-version and single-version being the two main approaches. BOHM and PWV employ multi-version storage to support high concurrency and non-blocking reads~\cite{sadoghi2014reducing,wang2024gria,BOHM,PWV}, while Calvin and Aria adopt single-version storage~\cite{aria,calvin}. The choice of storage model directly affects read/write performance and scheduling complexity~\cite{aria,calvin,kallman2008h,BOHM}.

Overall, although existing deterministic databases benefit from layered designs that enable efficient transaction execution, they still face some limitations, which affect the adaptability to dynamic workloads and efficiency under high contention and skew.

% Overall, although existing deterministic databases achieve efficient transaction execution through layered designs that encompass input, sequencing, scheduling, execution, and storage, they remain constrained by several inherent limitations:
% (1) input layers often require predefined logic transactions with fully known execution details, restricting support for interactive or ad-hoc workloads;
% (2) sequencing typically performs global ordering without leveraging transaction dependencies, missing early opportunities for conflict mitigation, and leaving dependency resolution to later stages;
% (3) scheduling mechanisms either rely on coarse-grained validation, causing excessive aborts under contention, or on fine-grained dependency graphs, which are expensive to construct for large batches;
% and (4) storage models are tightly coupled with concurrency control, where the choice between multi-version and single-version designs can lead to trade-offs between concurrency and execution simplicity.
% These limitations collectively hinder adaptability to dynamic workloads, reduce efficiency in high-contention and skewed scenarios.

\section{SYSTEM OVERVIEW}
\label{sec:overview}

In this section, we provide an overview of \system, including its core idea and architecture.

\noindent\textbf{Core Idea.} 
Deterministic databases aim to construct a globally consistent execution order while improving throughput and scalability. However, existing approaches face three major limitations. 
First, pessimistic approaches depend on pre-obtained read/write sets to construct dependency graphs, which often requires costly pre-execution and is impractical in real-world deployments. 
Second, optimistic approaches detect conflicts only after execution, leading to excessive aborts and wasted resources under high-contention workloads. 
Third, both approaches rely on ineffective reordering strategies. Optimistic systems use coarse rules that cause unnecessary aborts, while pessimistic systems adopt fine-grained graph construction that incurs high overhead. 
\system addresses these limitations by integrating three complementary techniques into a unified architecture: lightweight conflict prediction to remove the need for pre-obtained read/write sets, multi-version storage to reduce aborts under contention, and a lightweight reordering algorithm to enable efficient scheduling.

%%%%%建议这一段承上讨论一下体系结构设计的动机，也就是如何把前面的技术融合到一起 
%%%%%\textsc{\system}\footnote{change another font for terms}->\textbf{\system}
% hjf: 稍微再提3种技术，讨论一下体系结构设计的动机
% 小修改层描述，让每层都呼应前面的技术：Sequencing 层 → conflict prediction + \mtfs。Scheduling 层 → validation + conflict-aware filtering。Storage 层 → multi-version。
\noindent\textbf{Architecture.} 
The motivation of \system's architecture is to organize conflict prediction, reordering, and multi-versioning into a coherent workflow that supports efficient scheduling and reduces aborts under contention. Figure~\ref{fig:overview} shows the architecture of \system, which follows the general structure of deterministic databases while integrating \emph{three complementary techniques} for predictive scheduling.
(1) \emph{Input layer}: accepts general SQL transactions without requiring pre-obtained read/write sets, batching them for deterministic processing.  
(2) \emph{Sequencing layer}: applies conflict prediction and the \mtfs reordering algorithm to globally order and prune transactions, reducing unnecessary aborts.  
(3) \emph{Scheduling layer}: performs validation-based scheduling to ensure determinism while filtering transactions that still exhibit conflict risk.  
(4) \emph{Execution layer}: runs transactions in parallel on a consistent snapshot, buffering updates as new versions until commit.  
(5) \emph{Storage layer}: maintains multiple versions during execution, merges committed updates into persistent storage, and performs garbage collection at the end of each round.  

Next, we present the detailed workflow of \system, including five phases: prediction, execution, commit, fallback, and garbage collection.

\noindent\textbf{Prediction Phase.} This phase involves the input and sequencing layers. As illustrated in Figure~\ref{fig:overview}, \system first employs the Association Sum-Product Network (\aspn) model introduced in Section~\ref{sec:ACP} to capture data correlations and support fast conflict prediction. Then, \system applies the Matrix Two-pass Forward Scan (\mtfs) algorithm (Section~\ref{sec:reorder}) to reorder and filter transactions \emph{before execution}, deferring high-conflict ones to the fallback stage. This predictive scheduling reduces wasted work from executing transactions likely to be aborted.

\noindent\textbf{Execution Phase.} \system first generates a snapshot from the current persistent storage (Section~\ref{sec:optimization}). Each transaction executes on this snapshot, recording its updates in a local version chain as described in Section~\ref{sec:MultiVersion}. Since updates are buffered locally and not immediately applied to the database, all transactions read from a consistent snapshot throughout the execution phase.

\noindent\textbf{Commit Phase.} Deterministic validation with the \mtfs algorithm (Section~\ref{sec:reorder}) serves as a safeguard for transactions that have already executed. \mtfs constructs a global dependency matrix from the actual read/write sets of all transactions and selects a conflict-free subset for commit. Transactions outside this subset are deferred to the fallback stage, ensuring that any conflicts missed during the prediction phase are resolved.

\noindent\textbf{Fallback Phase.} During the fallback phase, \system executes transactions that failed to commit in earlier stages. The fallback strategy, detailed in Section~\ref{sec2:fallback}, ensures that these deferred transactions are re-executed under a conflict-free schedule.

\noindent\textbf{Garbage Collection Phase.} Multi-version-based optimization generates historical versions that must be reclaimed to prevent unbounded growth. \system adopts an epoch-based strategy~(Section~\ref{sec2:fallback}) to discard versions older than the latest active snapshot, ensuring both consistency and storage efficiency~\cite{nguyen2023survey,bottcher2019scalable,BOHM}. Finally, committed updates are merged into persistent storage to keep the database consistent and up to date.

\begin{figure}[!t]
    \includegraphics[width=\linewidth]{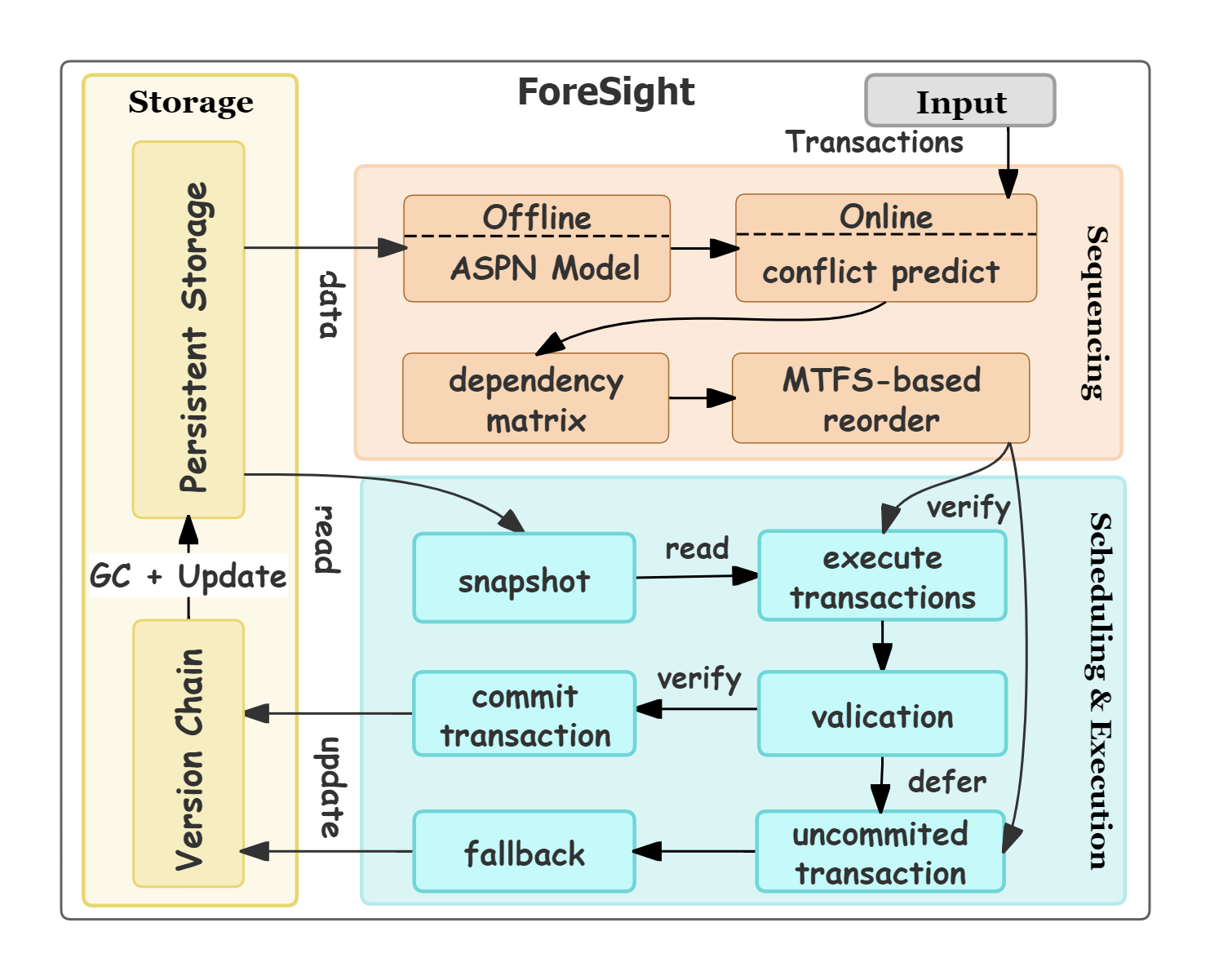}
    \caption{The Workflow of \system.}
    \label{fig:overview}
\end{figure}

\section{THE CONFLICT PREDICTION ALGORITHM}\label{sec:ACP}

%%%%%\footnote{discuss challenge and the basic idea here}\footnote{roadmap of this section}
Conflict detection is essential for deterministic execution, but existing methods rely on pre-executing transactions to extract read/write sets, incurring high overhead. To address this, we adopt predictive conflict analysis, which provides the sequencing layer with early dependency information for effective scheduling. In this section, we present Association Sum-Product Network (\aspn), a fast, lightweight, and incrementally maintainable model for conflict prediction. We formally define the prediction problem (Section~\ref{sec1:problem}), present the \aspn model (Section~\ref{sec1:aspn-model}), describe its prediction mechanism (Section~\ref{sec1:ACP}), and explain its incremental maintenance mechanism (Section~\ref{sec1:update}).

\subsection{Problem Definition}\label{sec1:problem}

We define transaction conflict prediction as a learning and inference problem that classifies dependencies between transactions. The goal is to determine whether any two transactions will incur read-write conflicts during concurrent execution, which guides the scheduler in performing safe and effective concurrency control.

\textit{Transaction Dependency.} Let $T = \{T_1, T_2, \ldots, T_n\}$ be a set of transactions, each defined as a tuple $(S_i, P_R^i, P_W^i)$, where $S_i$ denotes the target table, and $P_R^i$, $P_W^i$ represent read and write predicates, respectively. The read and write sets of $T_i$ are:
$$
R(T_i) = \{ t \in S_i \mid P_R^i(t) = \texttt{true} \}
,
W(T_i) = \{ t \in S_i \mid P_W^i(t) = \texttt{true} \}.
$$

The dependency between $T_i$ and $T_j$, where $T_i$ logically precedes $T_j$ in the scheduling timeline, is denoted as $D(T_i, T_j)$~\cite{thalheim2013dependencies}. It is classified as \textbf{RAW (Read-After-Write)} if $R(T_j) \cap W(T_i) \ne \emptyset$; \textbf{WAR (Write-After-Read)} if $W(T_j) \cap R(T_i) \ne \emptyset$; \textbf{WAW (Write-After-Write)} if $W(T_j) \cap W(T_i) \ne \emptyset$; and \textbf{IND (Independent)} if none of the above hold.

\textit{Conflict Prediction.}
To predict whether two transactions may conflict, we test for overlap in their read and write regions as defined earlier. For example, a RAW dependency occurs when the write region of transaction $T_i$ intersects with the read region of $T_j$. In practice, these regions are determined by the transactions’ access predicates, which are often expressed as multi-dimensional range queries in SQL, such as \emph{SELECT ... FROM S WHERE A <= a1 AND B <= b2 AND ...}. Such queries correspond to hyperrectangles in a high-dimensional space, making exact intersection checks costly and impractical during runtime.

To reduce this cost, we develop the \textbf{Association Sum-Product Network (\aspn)}, a compact probabilistic model that captures attribute correlations and efficiently estimates predicate overlap, to predict conflicts at runtime with low overhead.

\subsection{The \aspn Model}\label{sec1:aspn-model}

We build on the Sum-Product Network (SPN) as the foundation for \aspn. An SPN is a probabilistic model that factorizes a joint distribution using sum and product nodes~\cite{hilprecht13deepdb}, enabling linear-time inference by partitioning attributes under independence assumptions and decomposing heterogeneous data regions. The main advantage of SPN is its efficient and deterministic probability computation, since complex queries can be answered in time that scales with the network size, and once the structure and parameters are fixed, identical inputs always produce identical outputs. Such stability makes SPN appropriate as a deterministic inference component, where consistent results are essential to ensure determinism. However, SPN is not well-suited for database workloads, as real queries rarely satisfy the strong independence assumptions it relies on, and excessive partitioning often produces large structures. The factorize-split-sum-product network (FSPN)~\cite{flat} partially addresses these issues, but its design is primarily aimed at cardinality estimation rather than the predicate intersection required for conflict prediction.

To overcome these limitations, we propose the Association Sum-Product Network (\aspn), a tree-structured probabilistic model that applies conditional decomposition to groups of highly correlated attributes. This design preserves SPN’s efficiency and determinism while selectively modeling joint distributions for correlated subsets, yielding compact yet expressive representations. \aspn is particularly effective for conflict prediction, where the task is to determine whether two transaction predicates intersect a non-empty region of the data space. By capturing attribute correlations, it avoids the overestimation caused by independence assumptions, thereby reducing false positives and unnecessary aborts. This enables predictive scheduling that reduces transaction aborts and resource waste, ultimately improving throughput under high contention. Furthermore, since \aspn construction depends only on the given training samples and produces a fixed model once built, its inference yields consistent results. Given the same data, both the model and its predictions remain unchanged, making \aspn especially suitable for conflict prediction in deterministic databases.

\noindent\textbf{Structure Definition.}
Formally, given a set of attributes $A$ and record set $T$, \aspn models the joint distribution $\Pr_T(A)$ over $T$ using a tree structure. Each \aspn node $N$ is defined by three components: the target attributes $A_n \subseteq A$ to be modeled, the conditional attributes $C_n \subseteq C$ that provide the context, and the subset of training records $T_n \subseteq T$ associated with the node. According to their structural roles and statistical dependencies, \aspn nodes are categorized into four types:

\noindent\textbf{(1) Decomposition Node ($\oplus$):}
If the record distribution exhibits structural heterogeneity, $T_n$ is partitioned into clusters $T_1, \dots, T_k$, and local models are learned per cluster. The resulting distribution is expressed as a mixture: 
$\Pr_{T_n}(A_n \mid C_n) = \sum_{j=1}^{k} \Pr(T_j \mid T_n) \cdot \Pr_{T_j}(A_n \mid C_n)$. 

\noindent\textbf{(2) Independent Node ($\Box\times\Box$):}
If attributes in $A_n$ are mutually independent over $T_n$, then each attribute is modeled independently under the context $C_n$. The distribution is factorized as 
$\Pr_{T_n}(A_n \mid C_n) = \prod_{a \in A_n} \Pr_{T_n}(a \mid C_n)$.

\noindent\textbf{(3) Joint Node ($\Box\Box$):}
If attributes in $A_n$ are not independent, then a joint model is trained over $A_n$ under the context $C_n$, expressed as $\Pr_{T_n}(A_n \mid C_n)$.

\noindent\textbf{(4) Leaf Node ($\Box$):}
If $A_n$ contains only a single attribute $a$, a leaf node is created to model the marginal distribution of that attribute under the context $C_n$: $\Pr_{T_n}(a \mid C_n)$.

\subsection{Conflict Prediction Algorithm}\label{sec1:ACP}

We now describe the \aspn-based conflict prediction algorithm, which operates in two phases: (1) offline modeling, and (2) online inference.

% \begin{algorithm}[t]
%     \caption{\textsc{BuildASPN}$(A_n, C_n, T_n)$}
%     \label{alg:build-aspn}
%     \begin{algorithmic}[1]
%     \State \textbf{Input:} Target attributes $A_n$, conditional context $C_n$, training records $T_n$
%     \State \textbf{Output:} \aspn model for $T_n$.

%     \If{$|A_n| = 1$}
%         \State \Return \textbf{leaf node $(\Box)$} 
%     \EndIf

%     \State $D \gets \texttt{FindStrongSubset}(A_n, T_n)$

%     \If{$D = \emptyset$}
%         \For{each $a_i \in A_n$}
%             \State $N_i \gets \textsc{BuildASPN}(\{a_i\}, C_n, T_n)$
%         \EndFor
%         \State \Return \textbf{Independent node $(\Box\times\Box)$} with children $\{N_1, N_2, \ldots, N_{|A_n|}\}$
%     \ElsIf{$D = A_n$} 
%         \State \Return \textbf{Joint node $(\Box\Box)$} 
%     \Else 
%         \State $\{A_i, C_{i1}, \ldots, C_{im}\} \gets \texttt{Partition}(A_n, D)$
%         \For{$j = 1$ to $m$}
%             \State $(A_j, C_j) \gets (A_n \setminus A_i,\; C_n \cup C_{ij})$
%             \State $T_j \gets \{ t \in T_n \mid t[A_i] \in C_{ij} \}$
%             \State $N_j \gets \textsc{BuildASPN}(A_j, C_j, T_j)$
%         \EndFor
%         \State \Return \textbf{Decomposition node $(\oplus)$} with children $\{N_1, \ldots, N_m\}$
%     \EndIf

%     \end{algorithmic}
% \end{algorithm}
\begin{algorithm}[t]
  % \SetAlgoNoLine
  % \LinesNumbered
  \SetKwInOut{Input}{input}
  \SetKwInOut{Output}{output}

  \caption{\textsc{BuildASPN}$(A_n, C_n, T_n)$}
  \label{alg:build-aspn}

  \Input{Target attributes $A_n$, conditional context $C_n$, training records $T_n$}
  \Output{\aspn model for $T_n$}

  \BlankLine

  \If{$|A_n| = 1$}{
      \Return leaf node $(\Box)$
  }

  $D \gets \texttt{FindStrongSubset}(A_n, T_n)$

  \If{$D = \emptyset$}{
      \ForEach{$a_i \in A_n$}{
          $N_i \gets \textsc{BuildASPN}(\{a_i\}, C_n, T_n)$
      }
      \Return Independent node $(\Box \times \Box)$ with children $\{N_1, \ldots, N_{|A_n|}\}$
  }
  \ElseIf{$D = A_n$}{
      \Return Joint node $(\Box\Box)$
  }
  \Else{
      $\{A_i, C_{i1}, \ldots, C_{im}\} \gets \texttt{Partition}(A_n, D)$ \\
      \For{$j = 1$ \KwTo $m$}{
          $(A_j, C_j) \gets (A_n \setminus A_i,\; C_n \cup C_{ij})$ \\
          $T_j \gets \{ t \in T_n \mid t[A_i] \in C_{ij} \}$ \\
          $N_j \gets \textsc{BuildASPN}(A_j, C_j, T_j)$
      }
      \Return Decomposition node $(\oplus)$ with children $\{N_1, \ldots, N_m\}$
  }
\end{algorithm}

\begin{figure*}[!t]
    \centering
    \includegraphics[width=\linewidth]{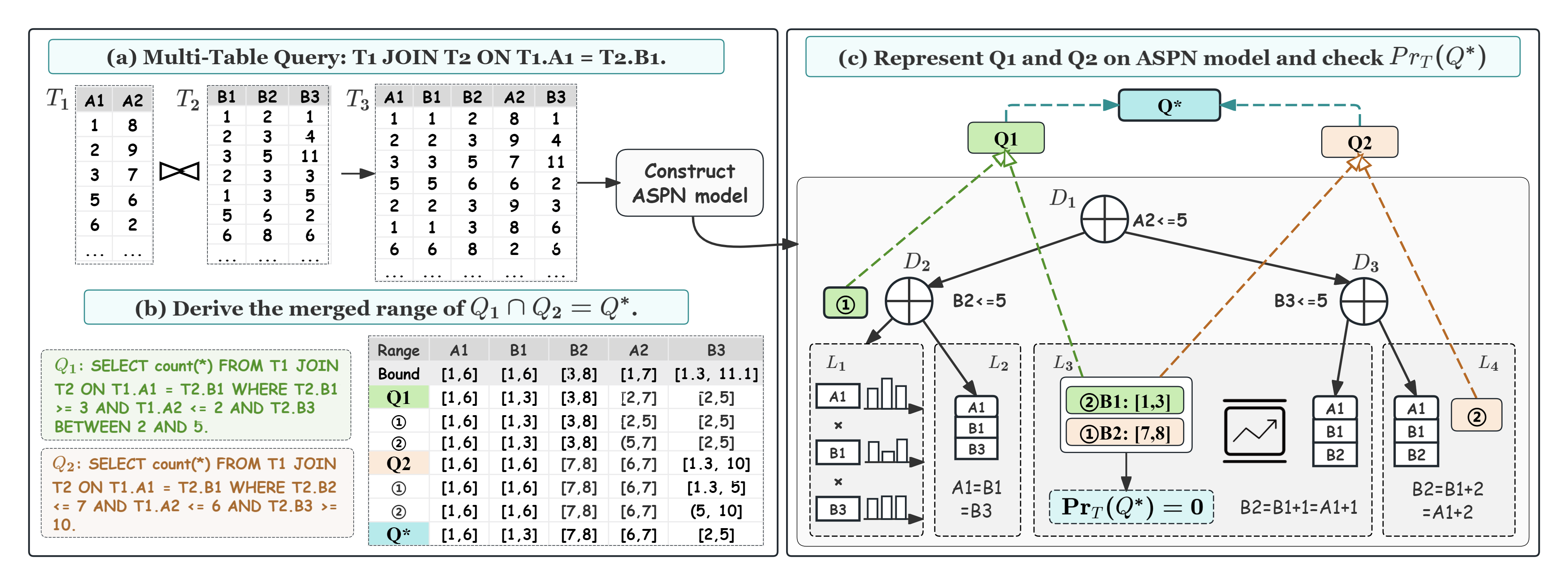}
    \caption{Example of \aspn-Based Conflict Prediction}
    \label{fig:ACP}
\end{figure*}

\noindent\textbf{Offline Modeling.} \aspn adopts a top-down recursive modeling strategy (Algorithm~\ref{alg:build-aspn}). At each node, the algorithm analyzes attribute dependencies and decides how to expand the model according to the following cases:

\textbf{Line 1–2 (Leaf Node):}
If the target set $A_n$ contains only one attribute, the recursion terminates and a leaf node $(\Box)$ is created. At this level, \aspn fits a simple density estimator (e.g., Gaussian mixtures, histograms, or other univariate models) to capture the marginal distribution $\Pr_{T_n}(a \mid C_n)$.

\textbf{Line 3 (Attribute Dependency Analysis):}
The algorithm invokes \texttt{FindStrongSubset} to identify a strongly dependent subset $D \subseteq A_n$ based on pairwise correlation scores (e.g., RDC~\cite{lopez2013randomized}). 

\textbf{Line 4–7 (Independent Node):}
If no strong dependencies are found (i.e., $D = \emptyset$), each attribute in $A_n$ is modeled independently. The algorithm recursively builds a univariate model for each $a_i \in A_n$ under the same context $C_n$ and returns an independent node $(\Box\times\Box)$ whose children are the independent attribute models.

\textbf{Line 8–9 (Joint Node):}
If all attributes are strongly correlated ($D=A_n$), a joint node $(\Box\Box)$ is created, which directly fits a joint probability estimator $\Pr_{T_n}(A_n \mid C_n)$ over the entire attribute set.

\textbf{Line 10–16 (Decomposition Node):}
If $D$ is a non-trivial subset of $A_n$, the algorithm selects a subset of attributes $A_i \subseteq A_n$ as the conditional context, and partitions $A_i$ into subsets ${C_{i1}, \dots, C_{im}}$ based on their values. For each subset $C_{ij}$, it recursively models the remaining attributes $A_j = A_n \setminus A_i$ under the updated context $C_j = C_n \cup C_{ij}$, with $T_j$ consisting of the records in $T_n$ that satisfy $t[A_i] \in C_{ij}$. Then $N_j \gets \textsc{BuildASPN}(A_j, C_j, T_j)$ is called to build the submodel. Finally, a decomposition node $(\oplus)$ is returned, aggregating all submodels as its children.

\textbf{Multi-table Modeling.}
When processing multi-table queries, \aspn incorporates join relationships to guide model construction. The process begins by building a join tree $J$, where each node corresponds to a table and each edge represents a join condition. For every join $(A, B)$ in $J$, the algorithm samples tuples from $A \Join B$ and computes correlations between attributes across the two tables. Table pairs with strong cross-table correlations are merged into a single modeling unit. This merging continues iteratively until no strongly correlated pairs remain, yielding several loosely connected subgraphs in which attributes are tightly coupled within each subgraph but weakly related across subgraphs. Each subgraph is then modeled independently using the \textsc{BuildASPN} procedure, producing a local \aspn model $F_T$ for that subgraph.

\textbf{Modeling Example.}
Figure~\ref{fig:ACP}(a) shows that $T_1$ and $T_2$ are merged into subgraph node $T_3$, covering attributes $A1$, $B1$, $B2$, $A2$, and $B3$, while Figure~\ref{fig:ACP}(c) illustrates the resulting \aspn structure. The modeling process first identifies $A2$ as nearly independent from the other attributes, splitting $T_3$ into $D_2$ and $D_3$ based on $A2 \leq 5$. Each subset is then examined for strongly dependent attribute groups. For $D_2$, $B2 \leq 5$ further divides the data into $L_1$ and $L_2$. $L_1$ contains independent attributes and is modeled using a multivariate histogram, $\Pr_{L_1}(W1) = \Pr_{L_1}(A1) \cdot \Pr_{L_1}(B1) \cdot \Pr_{L_1}(B3)$, while in $L_2$, $A1 = B1 = B3$ forms a joint variable $W2 = \{A1, B1, B3\}$ modeled as $\Pr_{L_2}(W2)$. For $D_3$, $B3$ is independent from the remaining attributes, producing $L_3$ and $L_4$: $L_3$ satisfies $B2 = B1 + 1 = A1 + 1$ and is modeled jointly as $\Pr_{L_3}({A1, B1, B2})$, whereas $L_4$ satisfies $B2 = B1 + 2 = A1 + 2$ and is modeled in the same way.

\textbf{Complexity Analysis.}
Let $n$ be the total number of nodes in the resulting \aspn, $d$ the number of decomposition nodes, $|A|$ the number of attributes, and $|T|$ the number of training records. The cost stems from three sources.
(1)~Each internal node invokes \texttt{FindStrongSubset} to evaluate attribute dependencies. This typically involves computing pairwise correlation scores (e.g., RDC) on $r$ sampled records, costing $O(|A|^2 r \log r)$ per node. 
(2)~For decomposition nodes, the algorithm partitions attributes into at most $b$ subsets, leading to $O(db)$ recursive calls, resulting in an additional cost of $O(d b)$.
(3)~The recursive structure ensures that each training record is only propagated down a single branch at each level. Therefore, the total data scan cost across all nodes is $O(|T|)$.
Overall, the training complexity is $O(n|A|^2 r \log r + db + |T|)$.

\noindent\textbf{Online Inference.} The goal of online inference is to estimate the probability that two queries $Q_1$ and $Q_2$ overlap in the data space. Let $Q^* = Q_1 \cap Q_2$ denote their intersection. This task reduces to computing $\Pr_T(Q^*)$. The process consists of two steps: (1) constructing the intersection region $Q^*$, and (2) recursively evaluating its probability using the \aspn model.

\textit{1. Intersection Construction.}
For each attribute $A_i$, the overlapping interval is computed as:
\begin{equation}
Q^*[A_i] =
\left[
  \max\left(L_i^{(1)}, L_i^{(2)}\right),\;
  \min\left(U_i^{(1)}, U_i^{(2)}\right)
\right].
\end{equation}

If $\max\!\left(L_i^{(1)}, L_i^{(2)}\right) > \min\!\left(U_i^{(1)}, U_i^{(2)}\right)$ for any attribute $A_i$, return $\Pr_T(Q^*) = 0$.

\textit{2. Recursive Evaluation.}
To compute $\Pr_T(Q^*)$, the query region $Q^*$ is recursively propagated through the \aspn structure. At each node, evaluation is performed according to its type. For decomposition nodes, $Q^*$ is partitioned based on conditional attributes and forwarded to each subcontext (e.g., $L_1, \dots, L_t$). At independent nodes, if $Q^*$ lies within the attribute domains and the attributes are mutually independent, the probability is computed as the product of marginal probabilities; otherwise, zero is returned. For joint nodes, if $Q^*$ intersects the joint domain, the joint probability is estimated from the fitted model; otherwise, the result is zero.

\textbf{Inference Example.}
In Figure~\ref{fig:ACP}(b), consider:

\emph{$Q_1$: SELECT count(*) FROM T1 JOIN T2 ON T1.A1 = T2.B1 WHERE T2.B1 <= 3 AND T1.A2 >= 2 AND T2.B3 BETWEEN 2 AND 5.}

\emph{$Q_2$: SELECT count(*) FROM T1 JOIN T2 ON T1.A1 = T2.B1 WHERE T2.B2 >= 7 AND T1.A2 >= 6 AND T2.B3 <= 10.}

The intersection of $Q_1$ and $Q_2$ is $Q^*$, which can be expressed as: \emph{A1:[1,6], B1: [1,3], B2: [7,8], A2: [6,7], B3: [2,5]}. $Q^*$ is projected onto the \aspn model, and each node is recursively evaluated. In $L_3$, due to the strong correlation (B2 = B1 + 1 = A1 + 1), the ranges \emph{B1: [1,3]} and \emph{B2: [7,8]} do not intersect; therefore, $\Pr_T(Q^*) = 0$.

\textbf{Complexity Analysis.}
Constructing the intersection $Q^*$ costs $O(|A|)$ since each attribute must be scanned once.
For recursive evaluation, independent nodes require only constant-time univariate estimation, and joint nodes similarly query lightweight joint estimators with constant cost.
Decomposition nodes contribute traversal overhead, one per node visited.
Let $i$, $j$, and $d$ denote the numbers of independent, joint, and decomposition nodes, respectively.
The total complexity is therefore $T(Q_1,Q_2) = O(|A| + i + j+d)$.

\subsection{Incremental Updates}\label{sec1:update}

When the table $T$ changes, the incremental update algorithm \texttt{Update} maintains the \aspn model by reusing the existing structure and adjusting only affected parameters.

In the single-table scenario, \texttt{Update} adopts a structure-preserving strategy, reusing most of the existing \aspn and applying local rebuilding only when necessary. Let $\Delta T$ denote the set of inserted or deleted records. \texttt{Update} traverses the \aspn top-down, performing localized parameter updates at each node: for decomposition nodes, $\Delta T$ is propagated to each corresponding sub-tree; for joint nodes, the system checks whether dependencies remain valid after applying $\Delta T$, and it updates the model parameters if they are preserved, otherwise it converts the node into a decomposition node and rebuilds the local structure via offline modeling; for independent nodes, the system verifies whether independence assumptions still hold, recursively updating marginal models when valid, and reconstructing the affected substructure when violated.

% In the single-table scenario, \texttt{Update} adopts a structure-preserving strategy, reusing most of the existing \aspn and applying local rebuilding only when necessary. Let $\Delta T$ denote the inserted or deleted records. The algorithm traverses the model top-down: decomposition nodes propagate $\Delta T$ to subtrees since they form a lossless partition; joint nodes re-validate correlations, updating parameters when they remain valid or converting into decomposition nodes and rebuilding via \texttt{Offline} otherwise; independent nodes verify independence assumptions, updating marginal models when preserved and reconstructing the substructure when violated.

This update strategy naturally extends to multi-table scenarios, where join operations introduce more complex dependency patterns. 
When a base table $T_i$ in a multi-table join $T = T_1 \Join T_2 \Join \dots \Join T_n$ is modified, the resulting changes to the join output $T$ impact three aspects: 
(1) new tuples generated by joining the updated records with other tables; 
(2) previously valid tuples that become invalid due to the join condition after the update; and 
(3) existing joined tuples whose scattering probabilities must be re-estimated.

For major schema modifications (e.g., attribute addition/deletion, new joins, or changes in join conditions), the system falls back to offline reconstruction to retrain the global \aspn, ensuring accuracy and robustness.

\section{Multi-Version-Based Optimization}
\label{sec:optimization}

Deterministic systems often suffer from frequent aborts under high contention due to aggressive validation rules that reject transactions with potential conflicts, leading to reduced throughput and wasted computation. To address this, \system introduces a multi-version-based optimization that relaxes commit constraints while preserving correctness, enabling more transactions to execute successfully. This section presents the design of our multi-version-based execution process (Section~\ref{sec:MultiVersion}), the fallback strategy for transactions that fail initial validation (Section~\ref{sec2:fallback}), and a formal proof of determinism and serializability (Section~\ref{sec:determinism-serializability}).

\subsection{Multi-Version-Based Execution Process}
\label{sec:MultiVersion}

\begin{figure*}[!t]
  \centering
  \includegraphics[width=\linewidth]{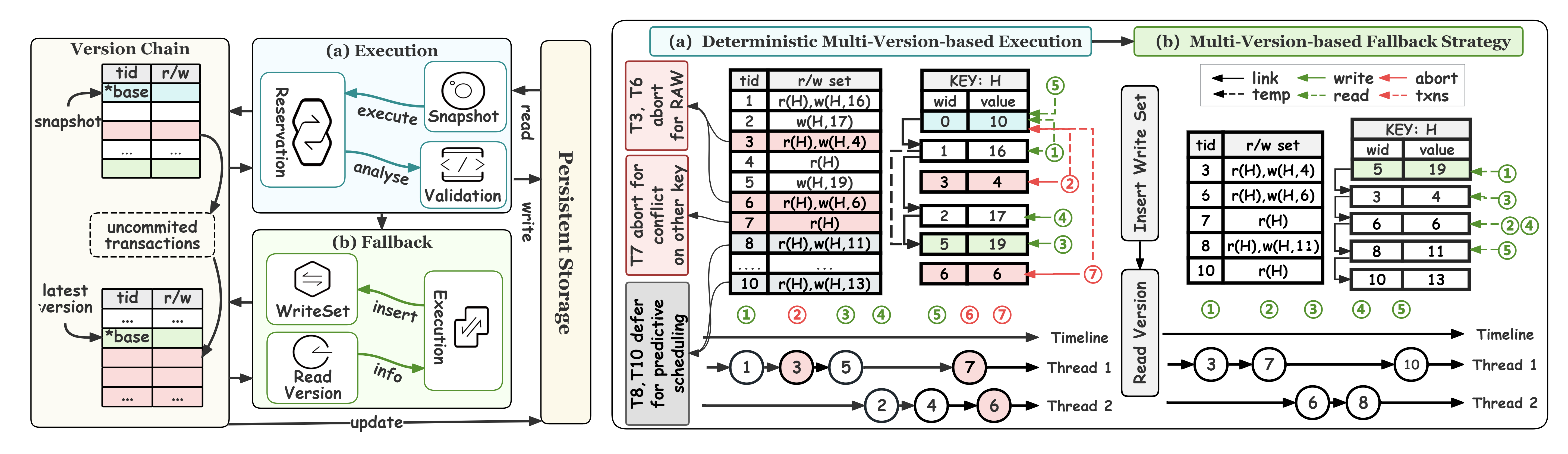}
  \caption{Multi-Version-Based Optimization in \system.}
  \label{fig:mvcc-scheduling}
\end{figure*}

\noindent\textbf{Background.} Deterministic systems like Aria~\cite{aria} typically adopt single-version storage and execute transactions on a consistent snapshot~\cite{chen2016fast}, followed by a validation phase. Without constructing a complete dependency graph, Aria validates a transaction for commit only if it has no WAW dependencies and no RAW dependencies. While this lightweight approach achieves high efficiency in low-contention environments, it suffers from high abort rates under high contention due to its limited ability to tolerate conflicts. In contrast, multi-version systems such as BOHM~\cite{BOHM} allow concurrent reads from historical versions and support more flexible dependency resolution, significantly reducing aborts under contention. However, these designs often incur higher complexity in scheduling, increasing runtime overhead.

\noindent\textbf{Key Idea.} The core challenge lies in striking a balance between the conflict tolerance of multi-version designs and the efficiency of lightweight deterministic execution. Existing approaches either suffer from excessive aborts or incur heavy coordination costs. To address this, we design a multi-version-based execution process that combines the simplicity of single-version scheduling with the robustness of multi-version concurrency control. This hybrid approach enables \system to sustain robust performance across both low- and high-contention workloads.

In \system, all transactions are executed on a consistent snapshot, avoiding synchronization during execution. Instead of modifying records in place, transactions buffer their updates as new versions, which remain invisible until commit. This deferred write strategy eliminates mid-execution conflicts and simplifies validation. The multi-version storage organizes each record as a version chain, with the head pointing to the latest visible version. Each version node stores the value, a pointer to its predecessor, and metadata such as a commit timestamp or transaction ID (TID). At the start of a batch, the head points to a base version representing the snapshot. During the commit phase, transactions append new version nodes for their writes, set timestamps to their TIDs, and atomically update the chain head. This organization ensures that write operations produce new versions rather than modifying records in place. Reads always access the snapshot-visible versions, enabling conflict-free execution in the execution phase.

Leveraging multi-versioning, \system adopts a relaxed validation strategy to enforce serializability. Unlike Aria, which rejects any transaction with WAW or RAW dependencies when reordering is disabled, \system allows WAW dependencies because multi-versioning ensures that writes never overwrite in-progress updates. This significantly increases commit opportunities under write-heavy workloads. \system only rejects transactions with RAW dependencies on prior transactions, as these can lead to inconsistencies if the transaction reads a value later overwritten by another transaction. This is formalized by the following rule:

\begin{ruledef}\label{rule1}
  \emph{A transaction is eligible to commit if it has no RAW dependencies on prior transactions.}
\end{ruledef}

\subsection{Fallback Strategy}
\label{sec2:fallback}

\noindent\textbf{Limitations of Lock-Based Fallback.} In ``execute-then-validate'' systems like Aria and DMUCCA~\cite{aria,DMUCCA}, high contention causes frequent aborts during validation. To mitigate this, the system defers aborted transactions to the next batch or switches to a lock-based fallback strategy~\cite{aria,calvin}, where transactions acquire read/write locks in TID order and are re-executed upon acquiring the necessary locks. Although the lock-based fallback strategy ensures serializability, the use of locks introduces spin and heavy overheads, especially under hot key contention. 

To address these limitations, \system adopts a multi-version-based fallback strategy. Fallback transactions run on the version chains generated at the end of the execution phase, which already include updates from all executed transactions. Instead of modifying records in place, each transaction appends its writes as new versions in the version chains, tagged with its TID and kept invisible until commit. When reading, a transaction retrieves the latest version whose timestamp precedes its own TID. Transactions in the fallback phase can execute in parallel, while their commits are finalized in deterministic timestamp order so that each incorporates all prior commits, preventing inconsistent reads. Following BOHM’s ``reads don’t block writes'' principle~\cite{BOHM, Silo}, this design eliminates global lock acquisition, avoids blocking among threads, and removes coordination overhead while preserving serializability.

Figure~\ref{fig:mvcc-scheduling} overviews the multi-version-based optimization in \system. A batch of transactions executes with snapshot reads guided by conflict-aware planning, while unsafe ones enter a fallback phase that pre-inserts writes, reads from version chains, and validates before commit. Figures~\ref{fig:mvcc-scheduling}(a)–(b) give an example: transactions $T_3$ and $T_6$ abort due to RAW/WAW conflicts on $T_1$ and $T_2$, while $T_7$  aborts on other keys, and $T_8$ and $T_{10}$ are early deferred by \aspn-based predictive scheduling. With the reordering strategy introduced in Section~\ref{sec:reorder}, the RAW dependency of $T_4$ is reordered into a WAR dependency, so $T_4$ can still commit. In the fallback phase, write transactions extend version chains sequentially, while reads traverse chains for visible versions. Multi-versioning ensures that even transactions with conflicts can proceed without blocking, with validation deferred to commit.
% Figure~\ref{fig:mvcc-scheduling} illustrates the multi-version-based optimization architecture of \system. Given a batch of transactions and their predicted read/write sets, the system first schedules them for parallel execution based on conflict-aware deterministic planning. Transactions that cannot be safely committed are diverted to the fallback phase, which inserts their write sets in advance, reads from a consistent snapshot, and performs validation before commit. Committed writes are appended to version chains, enabling concurrent reads and updates without blocking. 
% A detailed example of multi-version-based optimization execution and fallback process is shown in \autoref{sec:mvcc-example}.

\noindent\textbf{Garbage Collection.} \system uses epoch-based garbage collection to manage versioned storage efficiently. It tracks the minimum active snapshot timestamp to determine the visibility boundary and reclaims versions no longer needed~\cite{nguyen2023survey,bottcher2019scalable, BOHM}. Once fallback transactions are accomplished, old versions and those from aborted transactions are discarded~\cite{taft2020cockroachdb}. 

\subsection{Determinism and Serializability}
\label{sec:determinism-serializability}

%%%%%%下面是正式的证明，可以按照正式证明的格式来写，有proof字样，然后说分两个性质来证明
% hjf: 下面是正式的证明，分两个性质来证明, 加上了proof字样
We now prove \hyperref[theorem1]{Theorem 1} by establishing two properties, determinism and serializability.

\begin{theorem}\label{theorem1}
  \emph{Under \hyperref[rule1]{Rule 1}, \system ensures deterministic and serializable execution.}
\end{theorem}

\begin{proof}
\noindent\textbf{Determinism:} Given a batch of transactions and a database snapshot, \system assigns fixed transaction IDs and then uses \aspn to filter them, executing the filtered transactions on the same snapshot. Each transaction’s commit decision is based solely on its read/write sets and the presence of RAW dependencies with lower-TID transactions. Since \aspn is deterministic (Section~\ref{sec:ACP}), the same batch of transactions always produces the same filtering result, keeping the input consistent across replicas. Combined with the identical snapshot and validation rule (\hyperref[rule1]{Rule 1}), this guarantees that all replicas maintain deterministic execution.

\noindent\textbf{Serializability:} The following proof is by contradiction. Assume there exists a committed transaction sequence $\cdots \rightarrow T_i \rightarrow \cdots \rightarrow T_j \rightarrow \cdots$ and the result produced by \system is not serializable.
With concurrent execution and commit validation, two cases may cause non-serializability:
(1)~Transaction $T_j$ read a version written by $T_i$ (RAW), and  
(2)~Transaction $T_j$'s updates were overwritten by $T_i$'s updates (WAW). 
For (1), all transactions in a batch execute on the same consistent snapshot, and reads only access committed versions from prior rounds. Thus, $T_j$ cannot read a version written by $T_i$ within the same batch. If it did, it would violate snapshot isolation, contradicting \system's execution semantics.
For (2), \system employs multi-versioning, so each transaction appends a new version rather than overwriting existing data.  
Therefore, $T_i$ cannot overwrite $T_j$’s updates, and WAW does not cause write loss or inconsistency.
By \hyperref[rule1]{Rule 1}, transactions with RAW dependencies on prior transactions are rejected, and WAW dependencies are safely handled via version chains. 
Both cases lead to a contradiction. Hence, the execution must be conflict-serializable.
\end{proof}

\section{Reordering algorithm}
\label{sec:reorder}

Efficient transaction reordering is crucial for deterministic scheduling, but existing methods either incur high overhead or miss commit opportunities due to limited dependency awareness. To address this issue, we propose Matrix Two-pass Forward Scan (\mtfs), a lightweight algorithm that performs compact dependency analysis to produce a reordering that maximizes commit opportunities while preserving serializability. This section introduces a motivating example (Section~\ref{sec3:motivation}), describes the \mtfs algorithm (Section~\ref{sec3:mtfs}), and analyzes its correctness and effectiveness (Section~\ref{sec3:correct-analysis}).

\subsection{A Motivating Example}
\label{sec3:motivation}

\begin{figure}[!t]
  \centering
  \includegraphics[width=\linewidth]{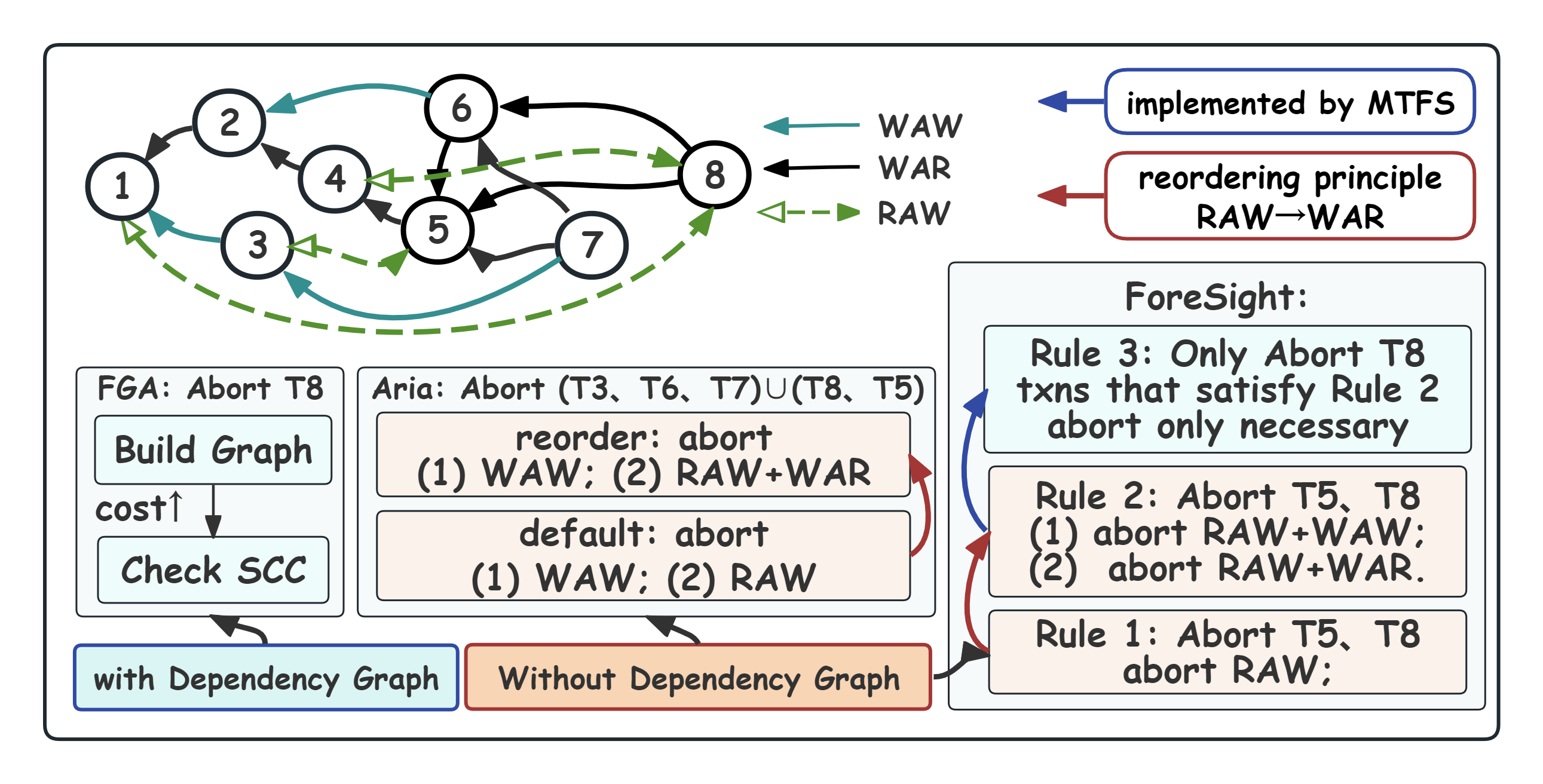}
  \caption{Illustration of two reordering strategies and the validation rules of \system}
  \label{fig:ProtocolCMP}
\end{figure}

The ordering of input transactions is the key to scheduling in deterministic databases, as it determines the commit order and affects the overall throughput~\cite{aria}. Figure~\ref{fig:ProtocolCMP} illustrates two representative reordering strategies and the validation rules of \system.

(1) \textbf{Without Dependency Graph:} A typical example Aria~\cite{aria}, which employs SSI-based cycle avoidance~\cite{chen2016fast} without building a full graph. When reordering is applied~\cite{Doppel}, transactions with RAW dependencies are reordered into WAR dependencies. Then, Aria only aborts transactions with WAW dependencies and those with both RAW and WAR dependencies, achieving higher performance.

(2) \textbf{With Dependency Graph:} 
The FVS-Greedy Algorithm (FGA)~\cite{fga} constructs the dependency graph after executing the batch, detects strongly connected components (SCCs)~\cite{nuutila1994finding}, and breaks cycles by aborting the transaction with the highest weight (based on in/out degrees) in each multi-node SCC. This approach is based on the notion of a feedback vertex set (FVS), which refers to the minimal set of vertices whose removal eliminates all cycles.

We presented a design where any transaction satisfying \hyperref[rule1]{Rule 1} can commit (Section~\ref{sec:optimization}). According to the reordering principle of Aria~\cite{aria}, the abort rate can be reduced by transforming RAW dependencies into WAR dependencies. Under snapshot reads, a RAW dependency leads to an abort because the read cannot see future writes. If reordered into a WAR dependency, the execution remains serializable, with the read finishing before the write. Applying such reordering enables \hyperref[rule1]{Rule 1} to be relaxed to \hyperref[rule2]{Rule 2}.

\begin{ruledef}\label{rule2}
  \emph{A transaction is not allowed to commit in either of the following cases:
  (1) It has RAW and WAW dependencies with preceding transactions;
  (2) It has RAW and WAR dependencies with preceding transactions.}
\end{ruledef}

To maximize commit opportunities, we further relax \hyperref[rule2]{Rule~2} to \hyperref[rule3]{Rule~3}, which selectively aborts the minimal set of transactions that would inevitably introduce cycles if committed. This refinement ensures serializability while minimizing the abort set, thereby allowing more transactions to commit.  

\begin{ruledef}\label{rule3}
\emph{Abort only the minimal set of transactions necessary to break inevitable cycles.}
\end{ruledef}

\begin{figure}[!t]
  \centering
  \includegraphics[width=\linewidth]{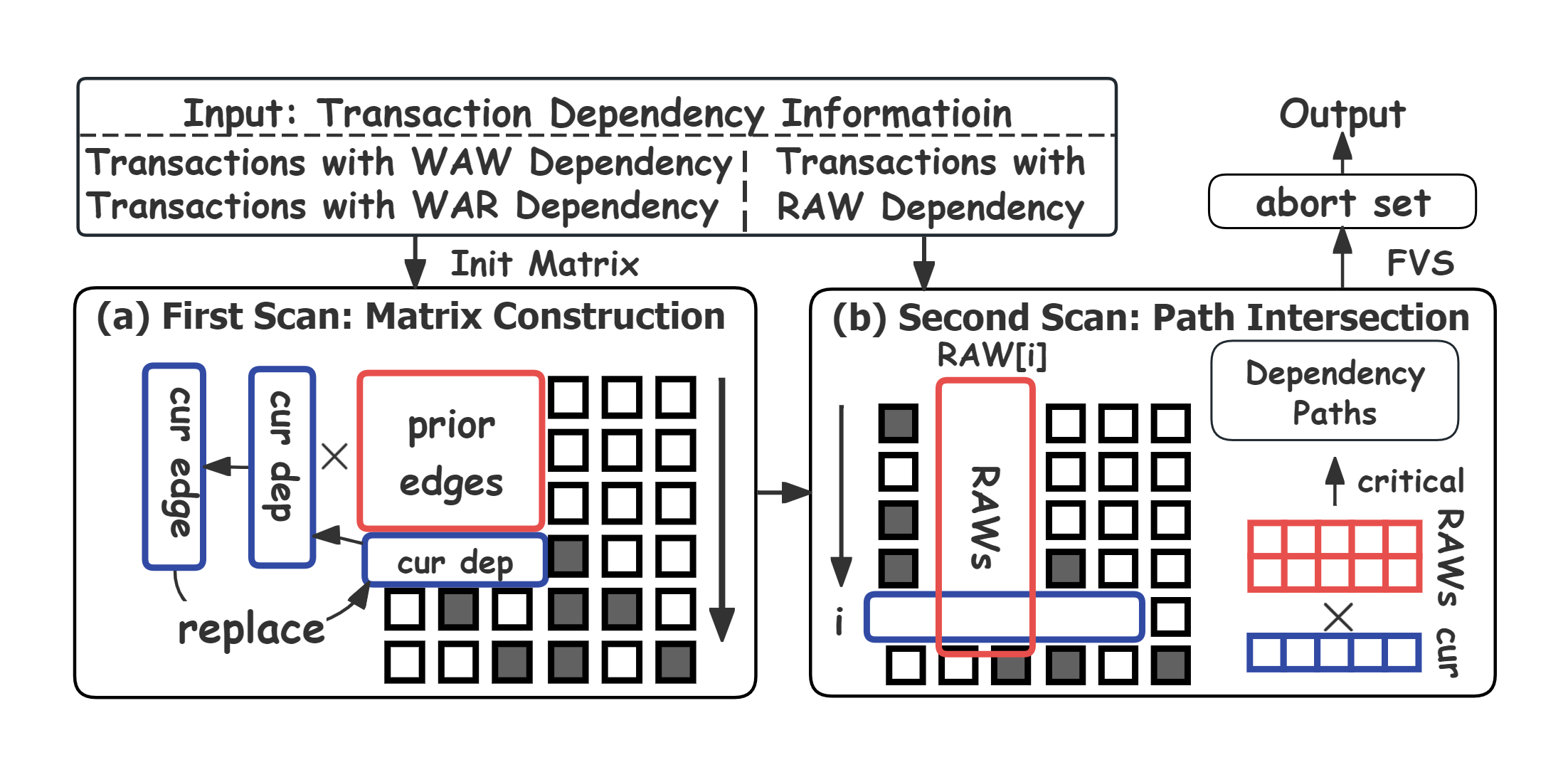}
  \caption{Illustration of Matrix Two-pass Forward Scan (\mtfs) algorithm.}
  \label{fig:mtfs}
\end{figure}

\subsection{Algorithm Design}
\label{sec3:mtfs}

To enable \hyperref[rule3]{Rule~3}, the system requires an efficient method to identify only those transactions that would inevitably form cycles, while maximizing commit opportunities.
Inspired by prior work such as SSFS (Stateful Single Forward Scan)~\cite{ssfs}, which constructs dependency graphs with a single forward scan by tracking transaction states. While SSFS achieves efficiency, it focuses on cycle detection rather than the selective abort minimization required by \hyperref[rule3]{Rule~3}. We develop the Matrix Two-pass Forward Scan (\mtfs) reordering algorithm, which extends the idea into two passes: one for building a global dependency matrix and another for producing a conflict-aware schedule. Figure~\ref{fig:mtfs} illustrates the process.
% and a detailed example is given in \autoref{sec:mtfs-example}.

\textbf{First Round Scan: Dependency Matrix Construction.} 
The objective of the first scan is to construct a \emph{dependency matrix} $M \in \mathbb{Z}^{N\times N}$, where $N$ denotes the number of transactions in the batch. This matrix serves as a compact representation of the dependency graph, with each entry $M[i][j]$ recording the number of distinct dependency paths from transaction $T_i$ to $T_j$, capturing both direct and indirect dependencies. It enables efficient analysis of dependency density for subsequent optimization.

Given a \emph{dependency dictionary} $\mathcal{D}$, where $\mathcal{D}[i]$ lists all transactions that $T_i$ directly depends on via WAR or WAW, the algorithm initializes $M$ by setting $M[i][i] = 1$ for self-dependency and $M[i][j] = 1$ for each $T_j \in \mathcal{D}[i]$ to represent direct dependencies.

Then, the algorithm propagates indirect dependencies in a row-by-row fashion.  For each $T_i$, if $M[i][j] > 0$ for some $j < i$, then all paths from $T_j$ are inherited by $T_i$, implemented by adding row $M[j]$ to row $M[i]$. Owing to the initialization, this row update is equivalent to multiplying the \emph{current dependency row} by the \emph{prior edges matrix}, allowing the transitive closure to be computed via efficient matrix operations. As shown in Figure~\ref{fig:mtfs}, the blue box (current row) accumulates dependency information from the red box (edges of predecessors), ensuring full coverage of indirect paths.

At the end of the scan, $M[i][j]=0$ indicates no path exists from $T_i$ to $T_j$; $M[i][j]=1$ indicates exactly one path; and $M[i][j]>1$ indicates multiple distinct paths, implying stronger dependencies.

\textbf{Second Round Scan: Reorder and Abort Set Extraction.}
After constructing the dependency matrix $M$, the second forward scan identifies a minimal set of transactions whose removal makes the dependency graph acyclic. This corresponds to finding a feedback vertex set~\cite{chen2008improved}. Since computing the minimum FVS is NP-complete~\cite{NP-complete}, our method uses a single-pass selection procedure that avoids the multiple full graph traversals required by degree-based selection methods~\cite{fga}.

Given $M$ and the dependency dictionary $RAW$, the algorithm outputs an abort set $S \subseteq {T_1,\dots,T_N}$. The method proceeds by identifying all dependency paths that involve RAW edges and selecting the transactions that occur most frequently on these paths.

First, the algorithm initializes a path set $\mathcal{P}$ to store all dependency paths. For each transaction $T_i$ with RAW dependencies (i.e., $RAW[i] \neq \emptyset$) and for each $T_j \in RAW[i]$, it identifies all transactions that appear on every path from $T_i$ to $T_j$ by taking the element-wise product of row $M[i]$ and column $M[:, j]$. Here, $M[i]$ marks all transactions reachable from $T_i$, while $M[:, j]$ marks all transactions that can reach $T_j$; nonzero entries in the product vector therefore correspond to transactions lying on all such paths. The resulting set of \emph{critical intermediates}, together with $T_i$ and $T_j$, forms a dependency path, which is added to $\mathcal{P}$. Repeating this process for all RAW edges accumulates the complete set of RAW-related paths that could participate in cycles.

Finally, the algorithm counts the frequency of each transaction across all paths in $\mathcal{P}$ and selects those with the highest frequencies into the FVS, forming the abort set $S$. Moreover, because $M$ encodes all dependency information in a compact form, both path extraction and FVS selection can be executed independently on different partitions of the batch~\cite{fleischer2000identifying,hong2013fast}, enabling scalable parallel execution.

\textbf{Complexity Analysis.}
The \mtfs algorithm consists of two main scans. In the first round, the algorithm performs one row-wise vector–matrix multiplication per transaction to propagate indirect dependencies; if each multiplication costs $O(m)$, the total cost is $O(N \cdot m)$. In the second round, a single vector multiplication per transaction detects critical transactions on dependency paths, also costing $O(N \cdot m)$. Finally, scanning the $l$ dependency paths to count transaction frequencies and select high-frequency candidates into the feedback vertex set costs $O(N \cdot l)$ using a linear-time method such as counting sort. Overall, the time complexity is $O(N \cdot (l + m))$, which is more scalable for large transaction batches than traditional methods requiring multiple passes or SCC detection~\cite{fga}.

\subsection{Correctness and Effectiveness Analysis}
\label{sec3:correct-analysis}

We now prove the correctness of our reordering algorithm and analyze its effectiveness in reducing abort rates.

\begin{figure}
  \centering
  \includegraphics[width=\linewidth]{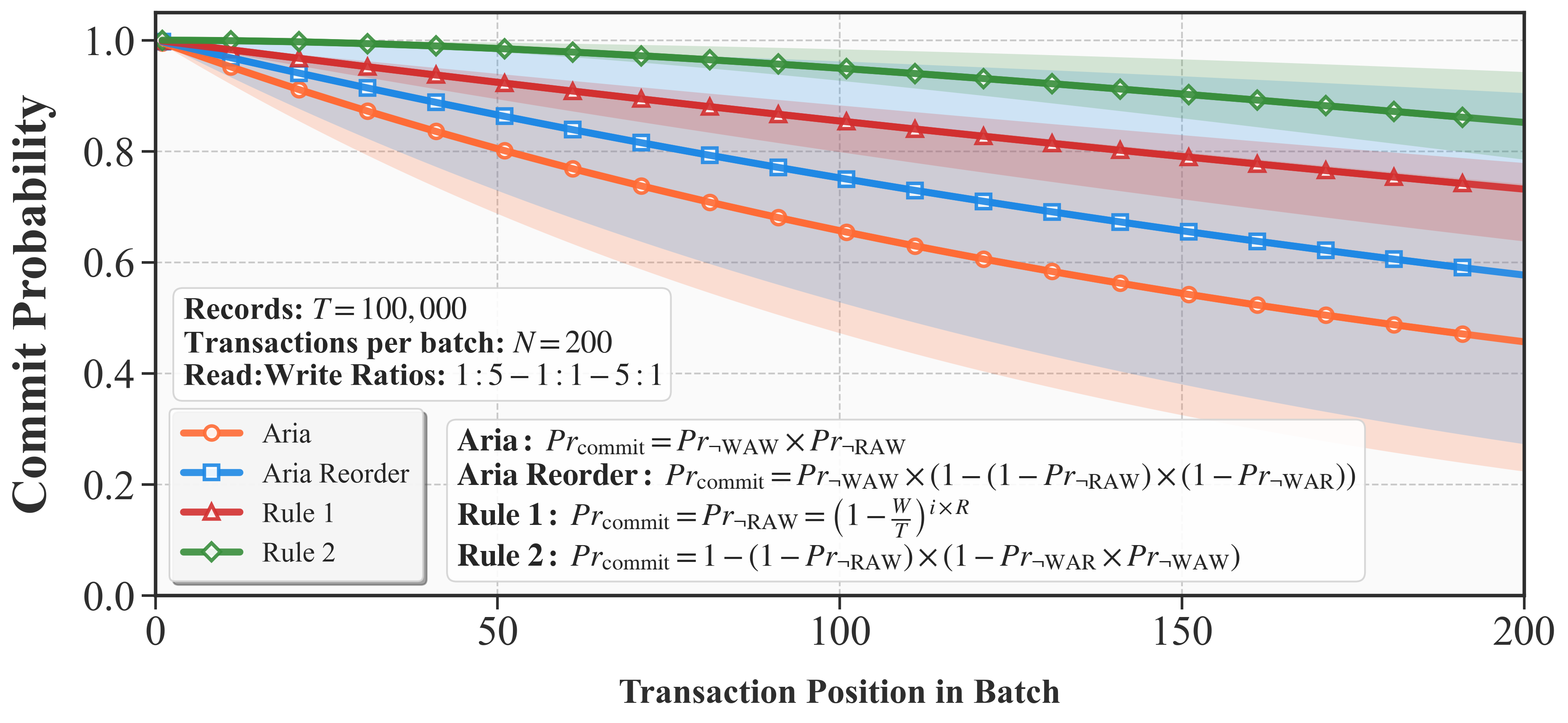}
  \caption{Commit probability under different read/write ratios (\system vs. Aria).}
  \label{fig:effective_reordering}
\end{figure}

\textbf{Correctness.}
The base reordering principle is captured in \hyperref[rule2]{Rule~2}, which transforms RAW edges into WAR edges whenever possible, placing the reader before the writer to avoid RAW violations~\cite{aria}. This transformation may add precedence edges; if a transaction already has WAW or WAR dependencies on earlier transactions, reordering can create cycles. To prevent this, \hyperref[rule3]{Rule~3} detects and aborts those transactions whose inclusion would inevitably form a cycle. The remaining transactions are either (i) WAR/WAW-only or (ii) reordered to remove RAW without forming cycles. After \hyperref[rule3]{Rule~3}, the dependency graph is acyclic, and a topological order yields a valid serializable schedule. Since dependency matrix $M$, path sets $\mathcal{P}$, and the abort set $S$ are computed deterministically per batch, all replicas derive the same order, ensuring determinism.

\textbf{Effectiveness Analysis.}
Consider a database with $T$ records and a batch of $N$ transactions, each with $R$ reads and $W$ writes over uniformly distributed keys. For transaction $T_i$, the probabilities of avoiding dependencies are:
$
\Pr[\neg\text{RAW}_i] = \left(1 - \frac{W}{T}\right)^{iR}, 
\Pr[\neg\text{WAR}_i] = \left(1 - \frac{R}{T}\right)^{iW}, 
\Pr[\neg\text{WAW}_i] = \left(1 - \frac{W}{T}\right)^{iW}
$.

Under \hyperref[rule1]{Rule~1}, \system aborts a transaction only if it has a RAW dependency, allowing all WAR/WAW-only transactions to commit. \hyperref[rule2]{Rule~2} further relaxes this by reordering, aborting only when RAW co-occurs with either WAW or WAR, yielding the commit probability:

\[
\Pr[T_i] = 1 - (1-\Pr[\neg\text{RAW}_i])\cdot(1-\Pr[\neg\text{WAR}_i]\cdot\Pr[\neg\text{WAW}_i])
\]

Furthermore, \hyperref[rule3]{Rule~3} maximizes commit opportunities by aborting only transactions that would inevitably form cycles. As shown in Figure~\ref{fig:effective_reordering}, the commit probability of \system is significantly higher than that of Aria, especially under contention. This demonstrates the effectiveness of the reordering algorithm in reducing abort rates and improving overall commit rates.

\section{Evaluation}
\label{sec:eval}

In this section, we evaluate the performance of \system with extensive experiments.

\subsection{Experimental Setup}

\begin{figure*}[!t]
  \centering
  \begin{subfigure}{0.32\linewidth}
      \centering
      \includegraphics[width=\linewidth]{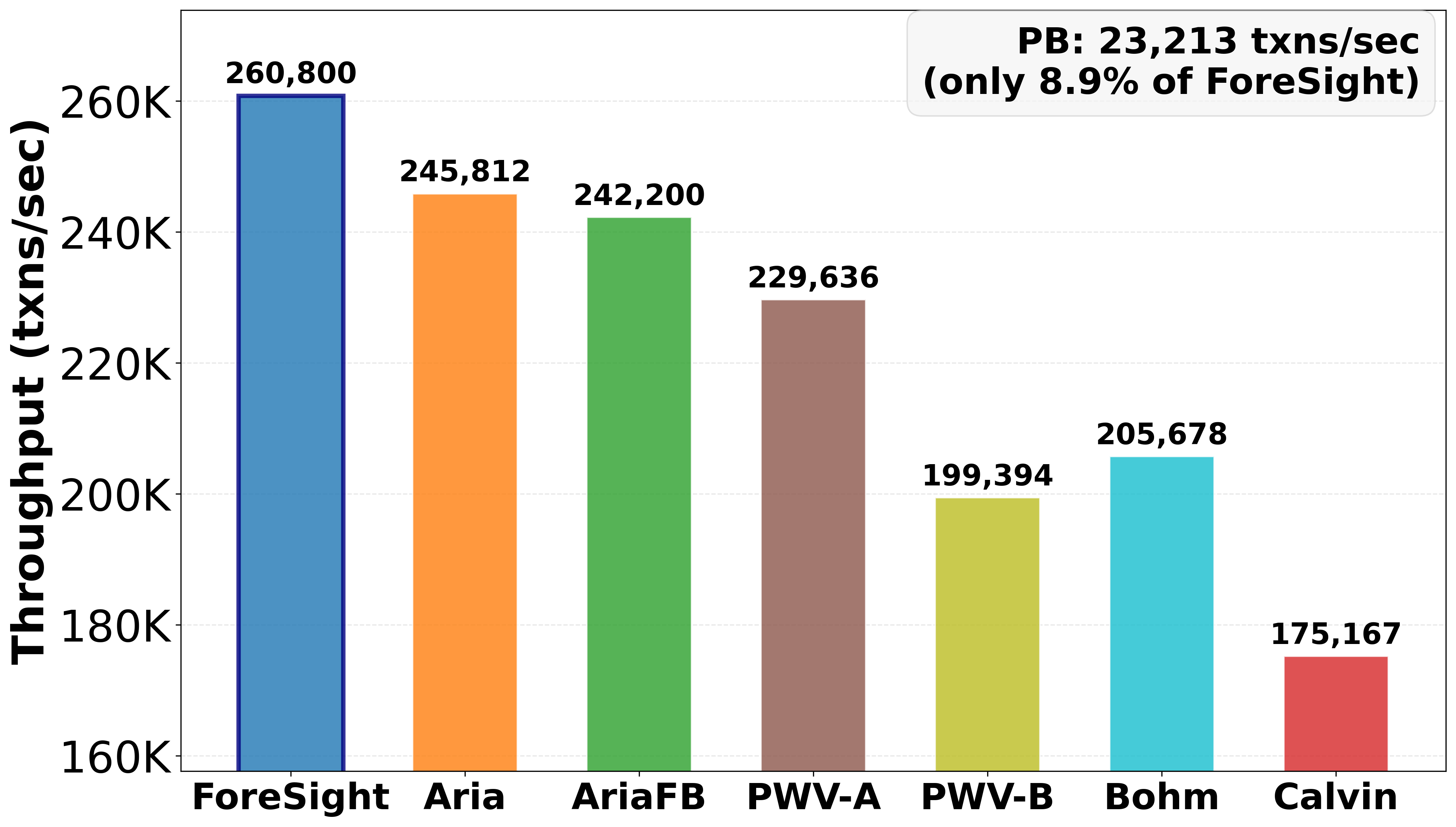}
      \caption{Throughput evaluation on YCSB benchmark}
      \label{fig:ycsb-comparison}
  \end{subfigure}
  \hfill
  \begin{subfigure}{0.32\linewidth}
      \centering
      \includegraphics[width=\linewidth]{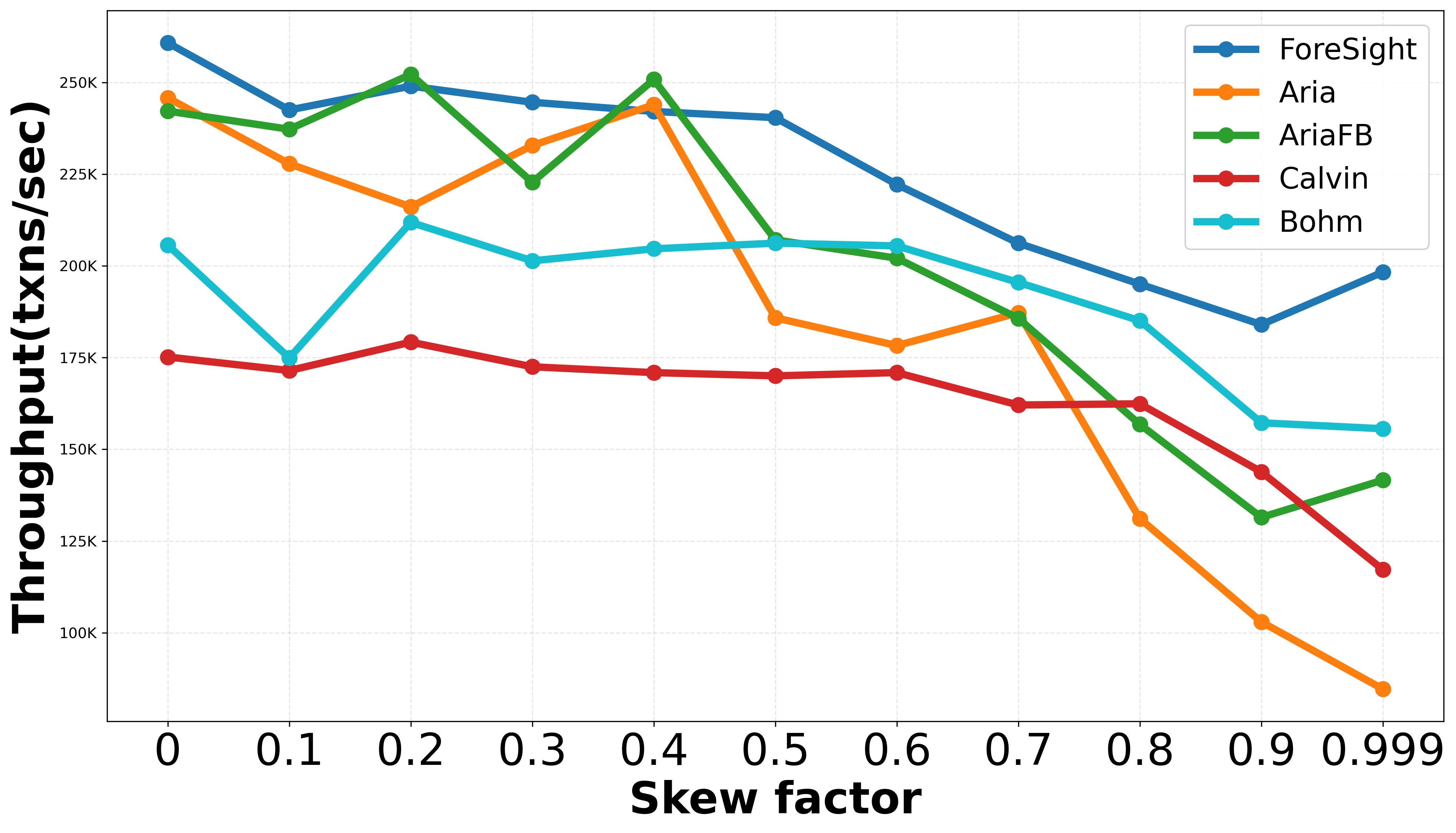}
      \caption{Throughput under YCSB with varying skew factors}
      \label{fig:skewnessPerformance}
  \end{subfigure}
  \hfill
  \begin{subfigure}{0.32\linewidth}
      \centering
      \includegraphics[width=\linewidth]{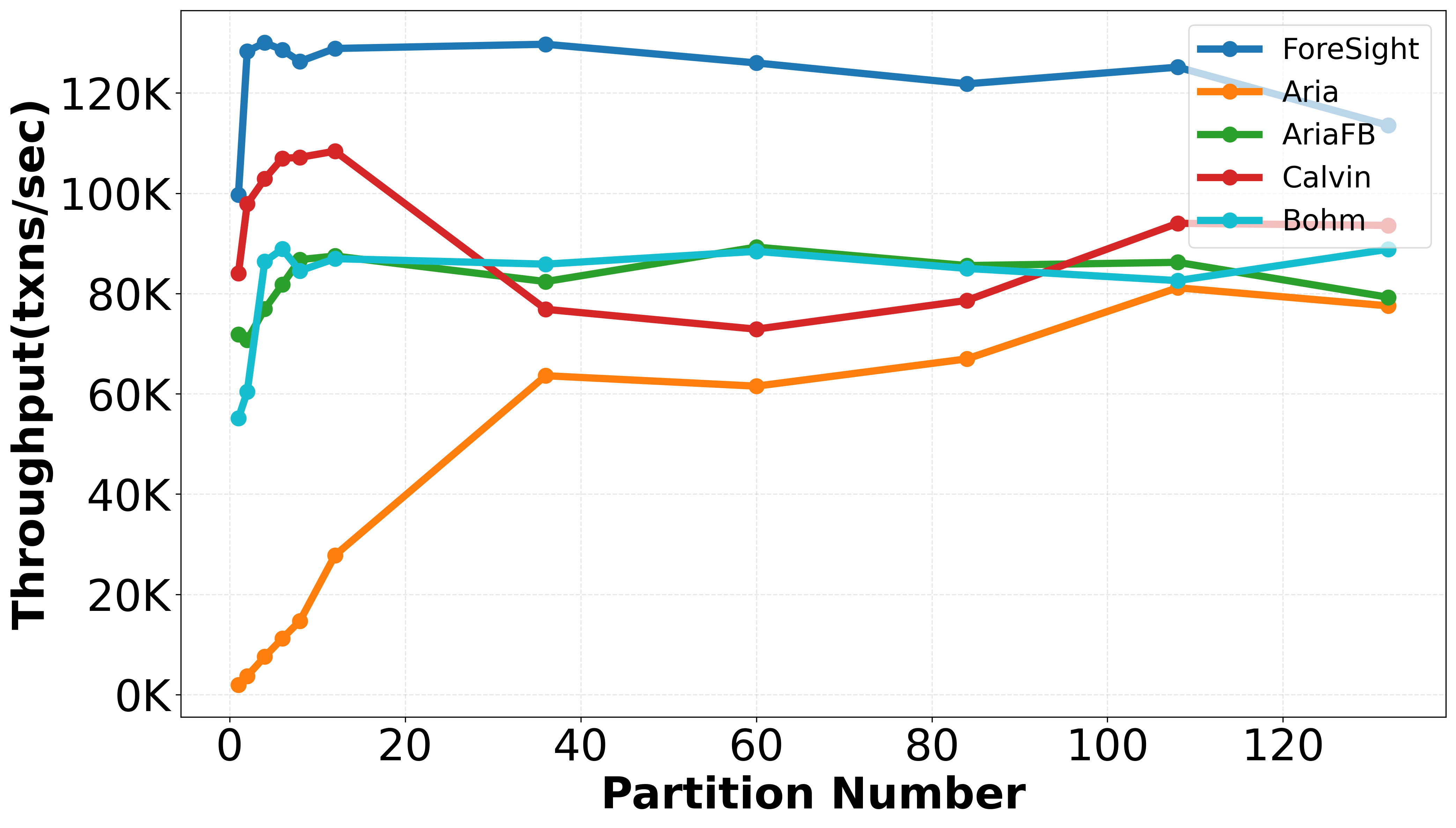}
      \caption{Throughput under TPC-C with varying numbers of partitions}
      \label{fig:partition-performance}
  \end{subfigure}
  
  \caption{Throughput and scalability under YCSB and TPC-C benchmarks}
  \label{fig:sample-num-robustness}
\end{figure*}

\textbf{Hardware settings.}  
All experiments are conducted in a virtualized environment with 8 vCPUs, 8 GB of RAM, and 60 GB of SSD per instance, running CentOS 7. The environment is hosted on a workstation with an Intel i7-12700H processor, 40 GB of RAM, and 2.5 TB of SSD storage.

\textbf{Benchmarks.}  
We evaluate our system using five benchmarks: TPCH~\cite{tpch} with a scale factor of 1, containing 22 tables and about 1 GB of data; IMDB~\cite{imdb}, a movie review dataset with 21 tables and 7 GB of data; TPC-C~\cite{tpcc}, an order-processing workload with ten warehouse partitions; YCSB~\cite{ycsb}, a single-table key-value schema with ten columns; and GAS~\cite{uci-gas}, a gas sensor dataset containing 3.84 million records with eight selected attributes. These benchmarks together cover scenarios ranging from analytical queries to high-concurrency transactional workloads.

\textbf{Workload Configurations.}  
We evaluate protocol performance and the effectiveness of reordering algorithms on TPC-C and YCSB, which support efficient pre-computation of read/write sets and are widely used in deterministic database studies. In TPC-C, transactions operate on ten warehouse partitions with ten operations per transaction, and the number of partitions equals the number of worker threads. YCSB includes two variants, $A$ with a balanced 50\% read/50\% update mix and $B$ with a read-heavy 95\% read/5\% update mix, both using a Zipfian distribution to simulate skew. Performance results are based on YCSB-A, except for PWV, where PWV-A/B denote runs on YCSB-A/B due to early write visibility effects. Each partition stores 40,000 keys, with batch sizes of 1,000 for YCSB and 500 for TPC-C. We also evaluate conflict prediction using 1,000 single-table queries on GAS and both 1,000 single-table and 1,000 multi-table queries on TPCH, IMDB, and TPC-C; all queries are randomly generated and contain at least one range predicate.

% All experiments use a deterministic workload generator to produce ordered transaction batches under serializability, with results averaged over ten runs. Batch sizes are 1,000 for YCSB and 500 for TPC-C.

\textbf{Protocols.} We evaluate the following systems: (1)~\textbf{\system:} A predictive-scheduling deterministic database. (2)~\textbf{Aria:} A deterministic reordering database with fallback strategy disabled~\cite{aria}. (3)~\textbf{AriaFB:} Aria employs fallback strategy~\cite{aria}. (4)~\textbf{BOHM:} A deterministic database based on dependency graphs~\cite{BOHM}.  (5)~\textbf{PWV:} A deterministic database supporting early write visibility~\cite{PWV}. (6)~\textbf{Calvin:} A deterministic database using ordered locking, configured with one lock manager thread~\cite{calvin}. (7)~\textbf{PB:} A traditional non-deterministic primary-backup database using strict two-phase locking (S2PL) and a NO-WAIT deadlock prevention strategy~\cite{soisalon1995partial}.

\textbf{Evaluation Metrics.} We evaluate the protocols mainly with three aspects: 

(1)~\textbf{Throughput} measures the number of processed transactions, indicating efficiency under varying contention and skew.  

(2)~\textbf{Commit Rate} measures the proportion of transactions committed, serving as a key indicator of how well reordering algorithms reduce conflicts under contention.

(3)~\textbf{Prediction Accuracy} evaluates conflict prediction as a classifier using standard metrics (accuracy, precision, recall):
\textit{Accuracy} — proportion of correct predictions, $(TP+TN)/(TP+TN+FP+FN)$, where TP (true positives) are correctly predicted conflicts, TN (true negatives) are correctly predicted non-conflicts, FP (false positives) are non-conflicts predicted as conflicts, and FN (false negatives) are conflicts predicted as non-conflicts; \textit{Precision} — proportion of predicted conflicts that are actual conflicts, $TP/(TP+FP)$; \textit{Recall} — proportion of actual conflicts correctly predicted, $TP/(TP+FN)$.

% These metrics together reflect both system performance and the quality of conflict prediction.

% \begin{figure*}[!t]
%   \centering
%   \begin{subfigure}{0.32\linewidth}
%       \centering
%       \includegraphics[width=\linewidth]{figures/exp/single_table_query.png}
%       \caption{Single-table Scenario}
%       \label{fig:sample-num-single}
%   \end{subfigure}
%   \hfill
%   \begin{subfigure}{0.32\linewidth}
%       \centering
%       \includegraphics[width=\linewidth]{figures/exp/multi_table_query.png}
%       \caption{Multi-table Scenario}
%       \label{fig:sample-num-multi}
%   \end{subfigure}
%   \hfill
%   \begin{subfigure}{0.32\linewidth}
%       \centering
%       \includegraphics[width=\linewidth]{figures/exp/prediction_time_combined.png}
%       \caption{Prediction Time }
%       \label{fig:sample-num-time}
%   \end{subfigure}
  
%   \caption{The Impact of Sample Size on Conflict Prediction Accuracy and Training Cost}
%   \label{fig:sample-num-robustness}
% \end{figure*}

\subsection{Evaluation of Protocol Performance}

In this section, we compare the performance of \system against state-of-the-art deterministic databases and the traditional primary-backup system PB, using the YCSB and TPC-C benchmarks.

\textbf{YCSB Results.} We first evaluate the performance of all systems on the YCSB benchmark, where experiments are conducted with 4 worker threads under an 80/20 read/write workload. Under a uniform distribution, most transactions span multiple partitions. We report the results in Figure~\ref{fig:sample-num-robustness}(a). \system achieves the highest throughput of 260,800 txns/sec, exceeding Aria by 1.06$\times$, AriaFB by 1.07$\times$, BOHM by 1.26$\times$, PWV-A by 1.13$\times$, PWV-B by 1.30$\times$, Calvin by 1.40$\times$, and PB by 11.23$\times$. The gain mainly comes from conflict prediction, which reduces aborts and re-executions even in YCSB’s low-contention setting. AriaFB is slightly slower than Aria because the fallback overhead outweighs its limited benefit in a low-abort environment. BOHM delivers 53\% of \system’s throughput, reflecting the cost of its version maintenance scheme. Calvin shows lower throughput because its locking protocol constrains concurrency. PB performs worst due to network round-trip delays and synchronous replication latency.

We then analyze the effect of the skew factor on throughput. As shown in Figure~\ref{fig:sample-num-robustness}(b), as the skew factor increases, the throughput of all systems declines. However, \system consistently maintains high throughput. Under a skew factor of 0.999, \system reduces contention-related aborts through conflict prediction, retaining 76\% of its no-skew throughput. In comparison, Aria retains only 34\%, and AriaFB 58\%, with AriaFB benefiting from its fallback strategy under skewed workloads. BOHM and Calvin also experience throughput degradation, retaining roughly 75\% and 66\% of their no-skew performance, respectively, though their absolute throughput remains consistently lower than \system across all skew levels.

\begin{figure}[!t]
  \centering
  \includegraphics[width=\linewidth]{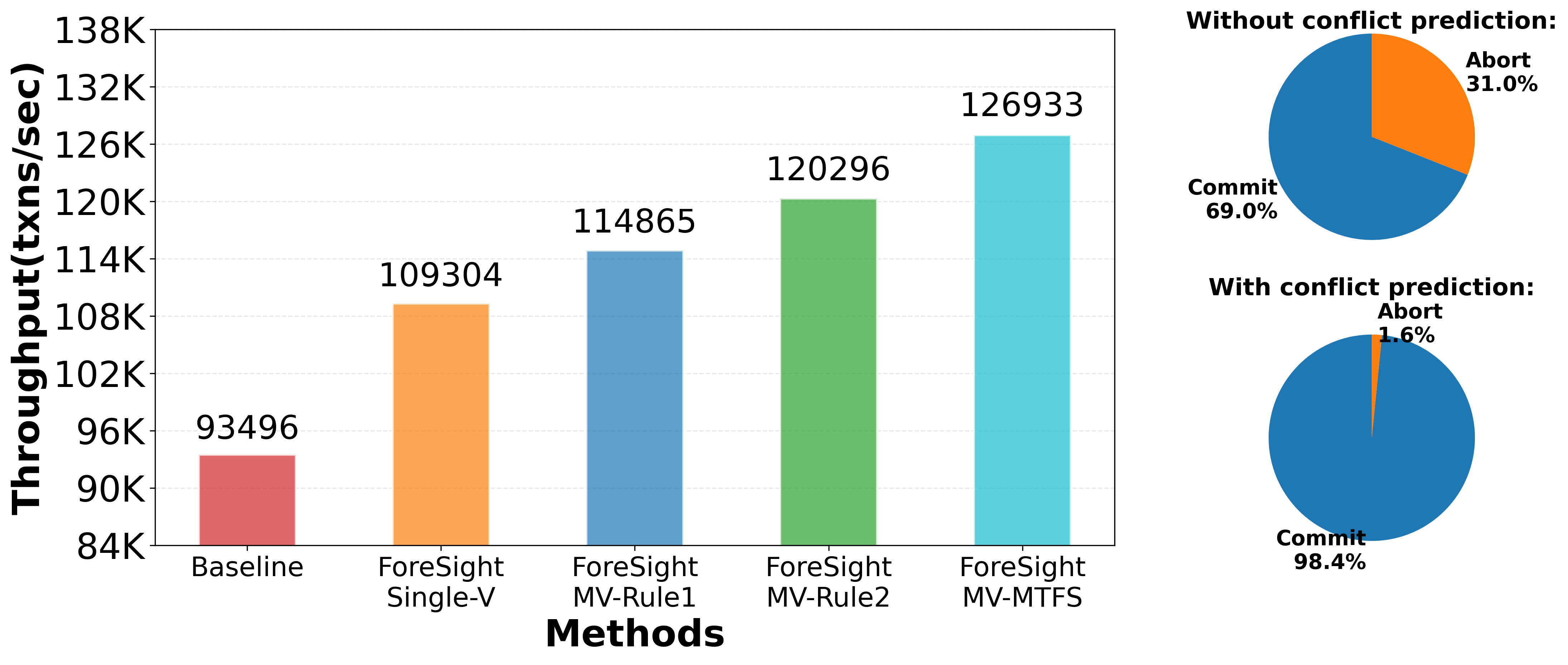}
  \caption{Ablation study of \system components on TPC-C}
  \label{fig:ablation}
\end{figure}

\textbf{TPC-C Results.} We next evaluate the performance of all systems on the TPC-C benchmark. As shown in Figure~\ref{fig:sample-num-robustness}(c), \system consistently maintains throughput above 120,000 txns/sec, significantly outperforming the other systems. Experimental results show that when the number of partitions exceeds 12, most deterministic databases plateau at around 80,000 txns/sec since each thread saturates its assigned partition under full parallelism. Aria’s throughput increases with the number of partitions due to reduced contention, eventually matching the other deterministic systems at 132 partitions, achieving 75\% of \system’s throughput. This is mainly because the TPC-C workload has certain fields that are frequently updated, which leads to high contention on these records. When the number of warehouses (partitions) is small, these updates cause more conflicts and lower performance. Systems with a fallback strategy, such as \system and AriaFB, handle conflicts more effectively, allowing their performance to remain stable even under high contention.

\textbf{The Ablation Study.} We conduct an ablation study in Figure~\ref{fig:ablation} to assess the impact of each core component in \system. The baseline system, which lacks all proposed techniques, achieves only 93{,}496 txns/sec. In contrast, the full \system configuration improves throughput to 126{,}933 txns/sec, a 1.36$\times$ improvement. \system Single-Version reaches 109{,}304 txns/sec (1.17$\times$), showing that predictive scheduling alone provides significant gains, and notably, enabling \system’s conflict prediction raises the commit rate from 0.690 to 0.984. Adding multi-version under \hyperref[rule1]{Rule 1} lifts throughput to 114{,}865 txns/sec (1.23$\times$), and relaxing it to \hyperref[rule2]{Rule 2} further raises performance to 120{,}296 txns/sec (1.29$\times$). Overall, throughput increases steadily as more components are enabled, confirming that predictive scheduling, multi-version with relaxed validation rule, and efficient reordering each contribute to performance and commit stability.

\subsection{Effectiveness of Conflict Prediction}

In this section, we evaluate the effectiveness of the conflict prediction algorithm described in Section~\ref{sec:ACP} under different conditions. There are many factors that can affect the performance of the conflict prediction algorithm, such as workload type, query complexity, and training sample size.

We evaluate the conflict prediction algorithm on TPCH and TPC-C workloads, covering both single-table and multi-table queries. The \aspn model is built on 10,000 sampled records and tested on 200 queries. As a baseline, we employ the learned abort prediction method proposed in~\cite{sheng2019scheduling}.
This method applies supervised learning by hashing transaction features and using logistic regression to predict conflicts~\cite{nasteski2017overview,lavalley2008logistic}. The model is trained on 800 labeled queries and tested on a separate set of 200 queries. Both models predict conflicts in batch mode. Evaluation metrics include accuracy~(acc), precision~(prec), recall, prediction time ($T_p$), and execution time ($T_e$), where $T_p$ measures the time required for conflict prediction, and $T_e$ measures the time required to obtain conflicts through actual query execution.

\begin{table}[!t]
  \centering
  \caption{Conflict Prediction Accuracy and Efficiency on TPCH and TPC-C Benchmarks}
  \small
  \resizebox{\linewidth}{!}{%
  \begin{tabular}{cccccccc}
  \toprule
  \textbf{Benchmark} & \textbf{Query} & \textbf{Model} & \textbf{Acc} & \textbf{Prec} & \textbf{Recall} & \textbf{$T_p$} & \textbf{$T_e$} \\
  \midrule
  \multirow{4}{*}{\textbf{TPCH}} 
    & \multirow{2}{*}{Single} 
      & \aspn & 0.999 & 1.000 & 0.994 & 0.437 & \multirow{2}{*}{19.726} \\
    &   & Baseline     & 0.723 & 0.735 & 0.899 & 0.120 &  \\
    & \multirow{2}{*}{Multi} 
      & \aspn & 0.913 & 0.884 & 0.970 & 0.400 & \multirow{2}{*}{97.128} \\
    &   & Baseline     & 0.723 & 0.729 & 0.725 & 0.184 &  \\
  \midrule
  \multirow{4}{*}{\textbf{TPC-C}} 
    & \multirow{2}{*}{Single} 
      & \aspn & 0.956 & 0.927 & 0.996 & 1.123 & \multirow{2}{*}{51.109} \\
    &   & Baseline     & 0.778 & 0.772 & 0.866 & 0.156 &  \\
    & \multirow{2}{*}{Multi} 
      & \aspn & 0.952 & 0.939 & 0.995 & 1.415 & \multirow{2}{*}{99.679} \\
    &   & Baseline     & 0.767 & 0.789 & 0.875 & 0.290 &  \\
  \bottomrule
  \end{tabular}
  } % end resizebox
  \label{tab:conflict-prediction}
\end{table}

\begin{figure}
  \centering
  \includegraphics[width=\linewidth]{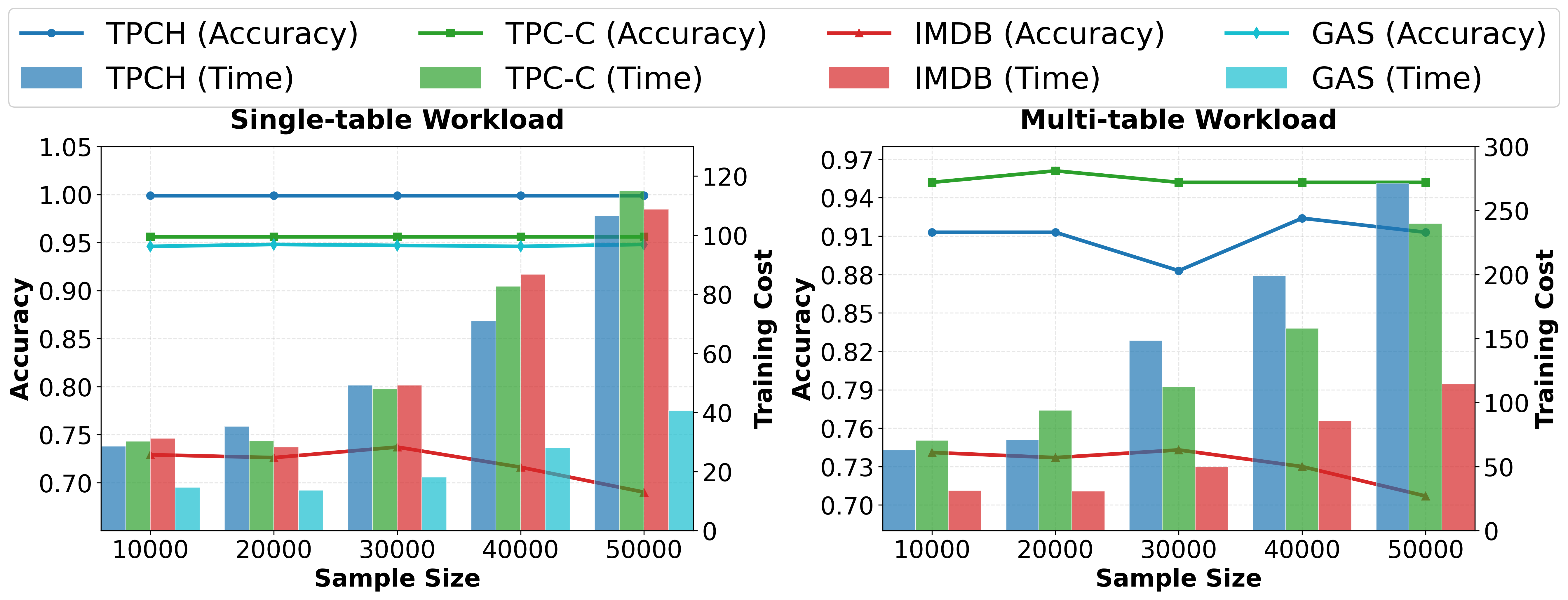}
  \caption{Impact of sample size on conflict prediction accuracy and training cost}
  \label{fig:sample-num-combined}
\end{figure}

As shown in Table~\ref{tab:conflict-prediction}, we observe that \aspn consistently outperforms baseline across all settings. For single-table queries, \aspn achieves near-perfect accuracy (99.9\% on TPCH and 95.6\% on TPC-C) and precision (100\% and 92.7\%, respectively), while maintaining high recall. In contrast, baseline exhibits significantly lower accuracy (72.3\%) and precision (73.5\%) on TPCH and fails to maintain high recall on multi-table queries. This indicates its limited capability to capture complex transactional patterns across joins or conditional logic. For multi-table queries, \aspn retains strong performance (91.3\% and 95.2\% accuracy; 88.4\% and 93.9\% precision on TPCH and TPC-C, respectively), while baseline degrades to near-random precision and recall, demonstrating its inability to generalize beyond simple access patterns. In terms of prediction time, \aspn achieves effective conflict prediction, taking 0.437 s and 1.123 s on single-table TPCH and TPC-C, and 0.400 s and 1.415 s on multi-table queries. Compared with query execution, which requires 19.7 s and 97.1 s for TPCH and 51.1 s and 96.7 s for TPC-C, the prediction overhead is negligible, accounting for only 2.2\% and 0.4\% of TPCH execution time and remaining below 3\% on TPC-C. This demonstrates that \aspn achieves high prediction accuracy while maintaining very low inference cost compared to query execution.

\begin{figure}
  \centering
  \includegraphics[width=0.95\linewidth]{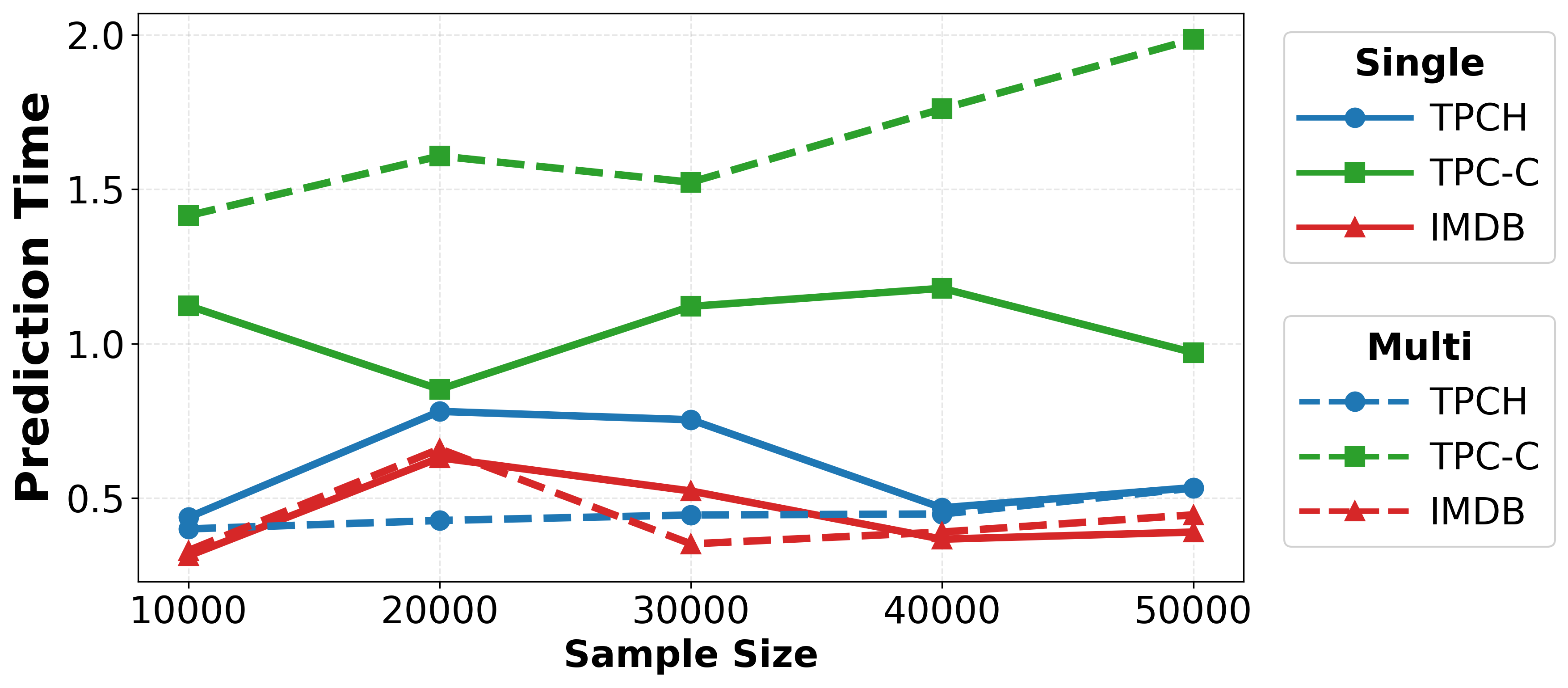}
  \caption{Impact of sample size on prediction time}
  \label{fig:sample-num-time}
\end{figure}

We next evaluate \aspn\ with training sample sizes ranging from 10{,}000 to 50{,}000 across both single-table and multi-table workloads. As shown in Figure~\ref{fig:sample-num-combined}, training costs increase with sample size and correlate with workload complexity, ranging from 14.6 to 40.6 seconds for GAS to 28.6 to 106.6 seconds for TPCH in single-table workloads. Precision remains largely stable ($\geq$94\%) across workloads, except for IMDB, which is lower ($\sim$75\%) due to its complex query patterns. Although prediction accuracy may vary across workloads, \aspn inference remains deterministic with identical samples producing identical outputs. Combined with the reorder-validation safeguard, this property guarantees execution consistency and determinism.
Figure~\ref{fig:sample-num-time} shows that prediction time grows roughly linearly with sample size but remains under 2 seconds for all workloads. This low overhead, combined with stable or improving precision, demonstrates that \aspn scales well to large datasets and remains suitable for real-time conflict prediction in high-throughput systems.

\subsection{Effectiveness of Deterministic Reordering}

In this section, we evaluate the effectiveness of the deterministic reordering algorithm \mtfs introduced in Section~\ref{sec:reorder}. The evaluation compares two implementations of \mtfs, the single-threaded baseline \system-S and the fully parallel implementation \system-M, against representative systems with reorder algorithms, including Aria~\cite{aria} and FGA~\cite{fga}. This evaluation allows us to highlight the benefits of predictive scheduling under different execution environments and workload characteristics.

We evaluate reordering effectiveness in Table~\ref{tab:reorder_effectiveness}. Under TPC-C, Aria reaches 5301.829 txns/sec with a 0.187 commit rate; FGA achieves 3625.079 txns/sec with 0.625. \system-S delivers 8187.369 txns/sec and a 0.861 commit rate, indicating fewer conflict-induced aborts. \system-M further exploits parallel execution to reach 17170.412 txns/sec—3.24$\times$ over Aria, 4.74$\times$ over FGA, and 2.1$\times$ over \system-S—without reducing the commit rate. On YCSB, where transactions are short and contention is low, the difference between \system-M and \system-S is small (2137.914 vs. 2175.957 txns/sec), showing that deterministic reordering yields larger benefits under high-contention workloads such as TPC-C.

\begin{table}[!t]
    \caption{Effectiveness of Reordering Algorithms}
    \label{tab:reorder_effectiveness}
    \small
    \begin{tabular}{lcccc}
    \toprule
    Protocol & \multicolumn{2}{c}{TPC-C} & \multicolumn{2}{c}{YCSB} \\
    \cmidrule(lr){2-3} \cmidrule(lr){4-5}
     & TPS & Commit rate & TPS & Commit rate \\
    \midrule
    Aria      & 5301.829 & 0.187 & 570.688  & 0.210 \\
    FGA       & 3625.079 & 0.625 & 1836.425 & 0.742 \\
    \system-S    & 8187.369 & 0.861 & 2175.957 & 0.707 \\
    \system-M & 17170.412 & 0.847 & 2177.914 & 0.710 \\
    \bottomrule
    \end{tabular}
\end{table}

Figure~\ref{reorder_rubustness} shows the robustness of the \mtfs reordering algorithm under varying workload skew in YCSB, comparing \system-M with Aria. The skew factor is increased from 0 (uniform distribution) to 0.999 (highly skewed) to test contention handling. With no skew, \system achieves 2177.914 txns/sec and a 0.710 commit rate, far exceeding Aria’s 570.688 txns/sec and 0.210. As skew grows, throughput drops for both systems due to more frequent conflicts. At high skew (skew factor = 0.9), \system reaches 133.474 txns/sec with a 0.315 commit rate, while Aria drops to 28.819 txns/sec and 0.017, showing that \mtfs reduces aborts and sustains higher throughput under skew.

\section{Related Work}\label{sec:related}

\system builds upon prior work in deterministic databases and transaction analysis for scheduling. This section reviews representative systems and techniques in these two areas, emphasizing their design principles and limitations.

% \subsection{Deterministic Database} 

\textit{Deterministic Database.} 
The concept of deterministic database systems can be traced back to the late 1990s~\cite{thomson2010case,narasimhai1999enforcing,jimenez2000deterministic,abadi2018overview}. 
The central idea is to enforce consistent execution across replicas by producing identical thread schedules through a deterministic scheduler~\cite{ren2012lightweight,ren2014evaluation,ren2015vll,faleiro2014lazy,DCC}. 
Representative systems such as Calvin~\cite{calvin}, BOHM~\cite{BOHM}, and PWV~\cite{PWV} adopt this approach by constructing dependency graphs or ordered locks from pre-obtained read/write sets. 
These techniques guarantee deterministic execution but restrict flexibility for workloads with data-dependent or interactive transactions. 
To avoid the overhead of dependency analysis and locking, Aria~\cite{aria} instead executes batches in parallel on consistent snapshots and deterministically validates them at commit time. 
Later work includes AriaER~\cite{ariaER}, which reduces aborts via early elimination of write-write conflicts, and DMUCCA~\cite{DMUCCA}, which reorders batches using access-pattern weights to improve throughput.
However, existing deterministic databases still face challenges such as the overhead of dependency analysis and conflict resolution, motivating the need for approaches that remain efficient under high contention.

% \subsection{Transaction Analysis}

\textit{Transaction Analysis.}
Transaction analysis techniques focus on studying transaction characteristics and dependencies to assist conflict management and scheduling. Early research shows that analyzing transaction structures can provide useful information for scheduling and improve throughput~\cite{calvin, BOHM}. 
Static approaches, such as Holistic Query~\cite{manjhi2009holistic} and QURO~\cite{yan2016leveraging}, merge operations or reorder queries based on dependency analysis. To mitigate the preprocessing cost of such methods, Zhang et al.~\cite{zhang2018performance} propose a lightweight static analysis that represents transactions and guides scheduling with reduced overhead.
Some approaches analyze dependencies at execution time. For example, the FGA algorithm~\cite{fga} groups transactions for pre-validation reordering, while DistDGCC~\cite{yao2018scaling} constructs distributed dependency graphs for cross-node transactions to improve conflict resolution. Systems such as Ic3~\cite{shasha1995transaction,wang2016scaling} dynamically maintain dependency graphs to adjust execution order in real time, providing strong serializability guarantees. In addition, learned models have been applied to transaction analysis for recognizing high-conflict transactions and guiding scheduling~\cite{sheng2019scheduling}.
While these techniques demonstrate the value of transaction analysis, challenges remain regarding overhead and adaptability. This motivates the development of lightweight methods that can support conflict avoidance and scheduling optimization in scalable systems.

\begin{figure}[!t]
    \includegraphics[width=\linewidth]{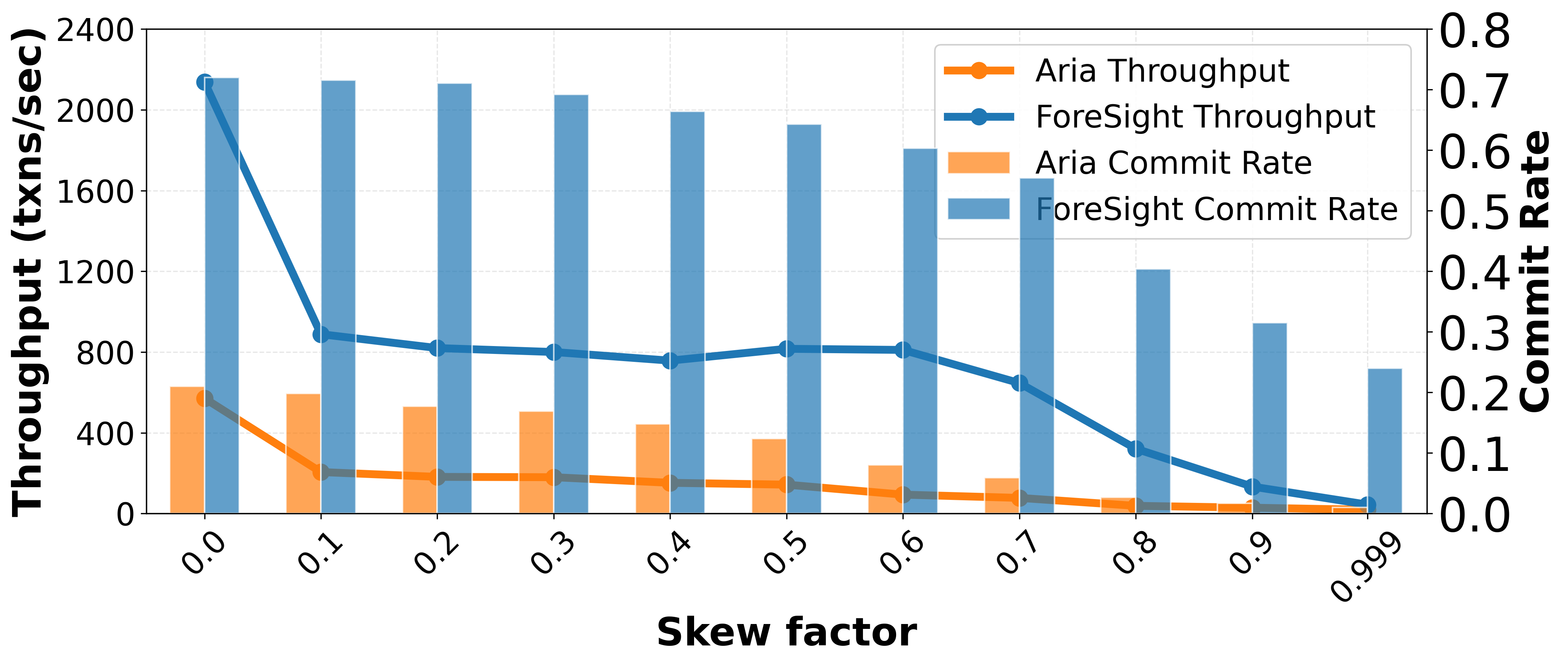}
    \caption{Robustness of reordering algorithms under skewed workloads}
    \label{reorder_rubustness}
  \end{figure}

\section{Conclusion}
\label{sec:conclusion}

In this paper, we introduced \system, a deterministic database system that achieves high performance through predictive, conflict-aware scheduling. By integrating lightweight conflict prediction, multi-version-based optimization, and efficient reordering, \system reduces aborts and improves throughput without requiring prior knowledge of transaction read/write sets. Extensive Evaluations on standard benchmarks confirm its effectiveness in handling high-contention workloads while maintaining deterministic guarantees. By introducing foresight into scheduling, \system offers a promising direction for scalable, reliable, and efficient deterministic databases. Additionally, future work will explore further optimizations in conflict prediction and scheduling strategies to enhance performance even more.

% \begin{acks}
% %   This work was supported by the National Natural Science Foundation of China (No.~62072138).
% \end{acks}
  
\clearpage

\bibliographystyle{ACM-Reference-Format}
\bibliography{sample}
  
% \clearpage
% \appendix

% \input{sections/appendix}

\end{document}